\documentclass[prx,reprint,floatfix,letterpaper,twocolumn,nofootinbib,showpacs,longbibliography,groupaddress]{revtex4-1}

\input{Qcircuit}

\usepackage[caption=false]{subfig}
\usepackage{graphicx} 
\usepackage{epstopdf}
\usepackage{array}
\usepackage{verbatim}
\usepackage{amsmath,amsfonts,amssymb,amscd}
\usepackage{amsthm}
\usepackage{tabularx}
\usepackage{stmaryrd}
\usepackage{enumerate}
\usepackage[ruled]{algorithm2e}
\usepackage{enumitem}
\usepackage{wasysym}
\usepackage[dvipsnames]{xcolor}
\usepackage[normalem]{ulem}
\usepackage[bottom]{footmisc}
\usepackage{bbm}

\usepackage{hyperref}
\hypersetup{
    bookmarksnumbered=true, 
    unicode=false, 
    pdfstartview={FitH}, 
    pdftitle={}, 
    pdfauthor={}, 
    pdfsubject={}, 
    pdfcreator={}, 
    pdfproducer={}, 
    pdfkeywords={}, 
    pdfnewwindow=true, 
    colorlinks=true, 
    linkcolor=blue, 
    citecolor=blue, 
    filecolor=blue, 
    urlcolor=blue 
}

\newcounter{ex}

\theoremstyle{plain}
\newtheorem{thm}{Theorem}

\newtheorem{lem}[thm]{Lemma}

\newtheorem{example}[ex]{Example}

\theoremstyle{definition}
\newtheorem{defn}[thm]{Definition}

\newcommand{\pvec}[1]{\vec{#1}\mkern2mu\vphantom{#1}}
\newcommand{\eq}[1]{(\hyperref[eq:#1]{\ref*{eq:#1}})}

\renewcommand{\sec}[1]{\hyperref[sec:#1]{Section~\ref*{sec:#1}}}
\newcommand{\thrm}[1]{\hyperref[thm:#1]{Theorem~\ref*{thm:#1}}}
\newcommand{\lemm}[1]{\hyperref[lemm:#1]{Lemma~\ref*{lemm:#1}}}
\newcommand{\prop}[1]{\hyperref[prop:#1]{Proposition~\ref*{prop:#1}}}
\newcommand{\corr}[1]{\hyperref[corr:#1]{Corollary~\ref*{corr:#1}}}
\newcommand{\fig}[1]{\hyperref[fig:#1]{Figure~\ref*{fig:#1}}}

\makeatletter
\def\legendre@dash#1#2{\hb@xt@#1{%
  \kern-#2\p@
  \cleaders\hbox{\kern.5\p@
    \vrule\@height.2\p@\@depth.2\p@\@width\p@
    \kern.5\p@}\hfil
  \kern-#2\p@
  }}
\def\@legendre#1#2#3#4#5{\mathopen{}\left(
  \sbox\z@{$\genfrac{}{}{0pt}{#1}{#3#4}{#3#5}$}%
  \dimen@=\wd\z@
  \kern-\p@\vcenter{\box0}\kern-\dimen@\vcenter{\legendre@dash\dimen@{#2}}\kern-\p@
  \right)\mathclose{}}
\newcommand\legendre[2]{\mathchoice
  {\@legendre{0}{1}{}{#1}{#2}}
  {\@legendre{1}{.5}{\vphantom{1}}{#1}{#2}}
  {\@legendre{2}{0}{\vphantom{1}}{#1}{#2}}
  {\@legendre{3}{0}{\vphantom{1}}{#1}{#2}}
}
\def\dlegendre{\@legendre{0}{1}{}}
\def\tlegendre{\@legendre{1}{0.5}{\vphantom{1}}}
\makeatother

\DeclareMathAlphabet{\matheu}{U}{eus}{m}{n}

\DeclareMathOperator{\tr}{tr}

\newcommand{\I}{{\mathbb I}}

\newcolumntype{L}[1]{>{\raggedright}p{#1}}
\newcolumntype{C}[1]{>{\centering}p{#1}}
\newcolumntype{R}[1]{>{\raggedleft}p{#1}}
\newcolumntype{D}{>{\centering\arraybackslash}X}

\definecolor{darkgreen}{rgb}{0,0.5,0}
\definecolor{darkblue}{rgb}{0,0,0.5}

\newcommand{\diag}[1]{\text{diag}(#1)}
\newcommand{\rank}[1]{\text{rank}(#1)}
\newcommand{\row}[1]{\text{row}(#1)}
\newcommand{\col}[1]{\text{col}(#1)}

\begin{document}
\title{Optimal quantum subsystem codes in 2-dimensions}
\author{Theodore J.~Yoder}
\affiliation{IBM T.J.~Watson Research Center}
\email{ted.yoder@ibm.com}
\begin{abstract}
Given any two classical codes with parameters $[n_1,k,d_1]$ and $[n_2,k,d_2]$, we show how to construct a quantum subsystem code in 2-dimensions with parameters $\llbracket N,K,D\rrbracket$ satisfying $N\le 2n_1n_2$, $K=k$, and $D=\min(d_1,d_2)$. These quantum codes are in the class of generalized Bacon-Shor codes introduced by Bravyi. We note that constructions of good classical codes can be used to construct quantum codes that saturate Bravyi's bound $KD=O(N)$ on the code parameters of 2-dimensional subsystem codes. One of these good constructions uses classical expander codes. This construction has the additional advantage of a linear time quantum decoder based on the classical Sipser-Spielman flip decoder. Finally, while the subsystem codes we create do not have asymptotic thresholds, we show how they can be gauge-fixed to certain hypergraph product codes that do.
\end{abstract}

\maketitle

\section{Introduction}
One of the perhaps more surprising facts to come out of quantum information theory is the close relation between classical and quantum error-correcting codes. Exemplary of this relation is the Calderbank-Shor-Steane (CSS) construction \cite{Calderbank1996,Steane1996}, which maps two classical codes (the first's dual contained in the second) to a quantum code. Important concepts in classical coding have analogous quantum concepts. For instance, a good family of classical $[n,k,d]$ or quantum $\llbracket n,k,d\rrbracket$ codes is one that asymptotically achieves constant rate $k/n$ and constant relative distance $d/n$. Using the CSS construction, one can draw on what is known classically to prove the existence of asymptotically good families of quantum codes \cite{Calderbank1996} and even construct them \cite{Ashikhmin2001,Chen2001}.

Because the classical codes input to the CSS construction must be related, it is sometimes difficult to use the CSS construction directly to make quantum codes with desirable properties. For example, the low-density parity check (LDPC) property, which can be defined for classical \cite{Gallager1962} or quantum \cite{MacKay2004,Tillich2014} codes alike, demands that every parity or stabilizer check involves a constant number of bits or qubits and every bit or qubit is involved in a constant number of checks. It is pointed out in \cite{MacKay2004} that one needs to use bad (i.e.~not good) classical LDPC codes to make quantum LDPC codes via the CSS construction, and that bad classical LDPC codes are uncommon, both because they are not worth studying if one is solely motivated by classical applications, but also because, asymptotically, most classical LDPC codes are actually good.

To easily create LDPC quantum codes, another method of converting classical codes to quantum ones has been developed. The hypergraph product \cite{Tillich2014} converts \emph{any} two classical codes to a quantum code. Notably, if the constituent classical codes are LDPC, so is the quantum code. The popular surface code is a special case, the hypergraph product of two classical repetition codes.

Yet, due to anticipated hardware limitations, it is common to place even more practical constraints on quantum codes beyond the LDPC condition. A popular demand is that parity checks are geometrically local in 2-dimensions so that it is unnecessary to interact qubits that are physically far apart in the plane. Bounds are known on the parameters $\llbracket N,K,D\rrbracket$ of 2-dimensional quantum codes of stabilizer subspace \cite{Bravyi2010} and subsystem \cite{Bravyi2011} varieties. The subspace bound $KD^2=O(N)$ is saturated constructively by the surface code \cite{Bravyi1998,Kitaev2003} and its relatives. The subsystem bound $KD=O(N)$ is known to be tight \cite{Bravyi2011}, but explicit constructions have heretofore been lacking.

Here, we establish another relation between classical and quantum codes. We show how to create an $\llbracket N,K,D\rrbracket$ quantum subsystem code that is local in 2-dimensions from any two classical codes with parameters $[n_1,k,d_1]$ and $[n_2,k,d_2]$ and prove that $N\le2n_1n_2$, $K=k$, and $D=\min(d_1,d_2)$. The quantum code belongs to the class of generalized Bacon-Shor codes introduced by Bravyi \cite{Bravyi2011}, a class we therefore refer to simply as Bravyi-Bacon-Shor codes. One can recover the traditional Bacon-Shor code \cite{Bacon2006,Aliferis2007} from our construction by starting with two classical repetition codes.

Bravyi-Bacon-Shor codes created this way have two important properties related to the constituent classical codes. First, if the classical codes are good, then the Bravyi-Bacon-Shor codes saturate the 2-dimensional subsystem code bound $KD=O(N)$. Second, decoders for the classical codes can be used to decode the quantum code. If the classical codes are LDPC and their decoders take linear time (in the size of the classical code), then the quantum decoding, including both data and measurement errors, takes linear time (in the size of the quantum code). Handling measurement errors in the quantum setting requires that the classical decoders handle errors in calculations of the parity checks. Though this is not a standard model in classical error-correction, the Sipser-Spielman flip decoder for expander codes \cite{Sipser1996} does apply to this situation \cite{Spielman1996}. Interestingly, decoding quantum expander codes \cite{Leverrier2015}, a kind of hypergraph product code, also employs what is in some sense a quantum version of this classical flip decoder \cite{Fawzi2018,Fawzi2018b,Grospellier2018}.

Finally, we show how to gauge-fix Bravyi-Bacon-Shor codes. This is the process of moving encoded data from a quantum subsystem code into a related subspace code. For instance, the Bacon-Shor code can be gauge-fixed to the surface code \cite{Li2018}. Thus, as a generalization of Bacon-Shor codes, Bravyi-Bacon-Shor codes should gauge-fix to a generalization of the surface code. This is indeed the case. We show that a Bravyi-Bacon-Shor code can be gauge-fixed into certain hypergraph product codes -- the hypergraph product of a classical repetition code and (either) one of the classical codes used to build the Bravyi-Bacon-Shor code. This reveals Bravyi-Bacon-Shor codes as a kind of \emph{subsystem} hypergraph product code.

In Section~\ref{sec:code_background} we review the codes we will be discussing and establish notation. In Section~\ref{sec:constructing_and_decoding}, we provide our construction of Bravyi-Bacon-Shor codes from classical codes and show how to decode them. In Section~\ref{sec:gauge_fixing}, we gauge-fix Bravyi-Bacon-Shor codes to hypergraph product codes. Section~\ref{sec:discussion} concludes.

\section{Code Background}\label{sec:code_background}
In this section, we review the codes that play a major role in the paper. These are (a) classical codes, including transpose and LDPC codes, (b) quantum subsystem codes, (c,d) two versions of Bravyi-Bacon-Shor codes, and (e) hypergraph product codes.

\subsection{Classical codes and their transposes}\label{sec:classical}
In this paper, we use ``classical code" to mean a classical linear code. A linear code $\mathcal{C}$ is a subset of the set of length-$n$ bit strings $\mathcal{C}\subseteq\mathbb{F}_2^n$ and can be defined by a parity check matrix $H\in\mathbb{F}_2^{m\times n}$ by setting $\mathcal{C}=\ker(H)$. This means that $w\in\mathcal{C}$ if and only if $Hw=0$. Notice, however, that $H$ itself is not unique.

The number of encoded bits $k=\dim\mathcal{C}$ is related to the rank of $H$ by the rank-nullity theorem
\begin{equation}
k=n-\rank{H}.
\end{equation}
Gaussian elimination can be used to find a basis for the kernel of $H$. This basis can be arranged as the rows of a generating matrix $G\in\mathbb{F}_2^{k\times n}$ satisfying $\rank{G}=k$ and $HG^T=0$. Of course, any $G'=QG$ for full-rank matrix $Q\in\mathbb{F}_2^{k\times k}$ is an equally valid generating matrix.

The distance of the code is the minimum (Hamming) weight of a nonzero vector in $\mathcal{C}$. That is,
\begin{equation}
d=\min\{|\vec w|>0:\vec w\in\mathcal{C}\}.
\end{equation}
Code parameters of $\mathcal{C}$ are collected in the tuple notation $[n,k,d]$.

Although not part of traditional classical coding theory, the ``transposes" of a classical code will be important for defining hypergraph product codes in Section~\ref{sec:HGP}. A code $\mathcal{C}^T$ is a transpose of $\mathcal{C}$ provided a parity check matrix $H$ exists so that $\mathcal{C}=\ker(H)$ and $\mathcal{C}^T=\ker(H^T)$. Let us say that $H\in\mathbb{F}_2^{n^T\times n}$, where $T$ modifying a scalar (like $n$) is to be treated as a superscript (not the transpose). Thus, $\mathcal{C}^T$ is another linear code with parameters $[n^T,k^T,d^T]$. Codewords in $\mathcal{C}^T$ represent redundancy (linear dependencies) between parity checks, the rows of $H$. Indeed, by the rank-nullity theorem and the fact that the column rank and row rank of a matrix are equal,
\begin{equation}\label{eq:transpose_relation}
n-k=n^T-k^T.
\end{equation}
If $H$ were full rank (i.e.~no check redundancy), $n^T=n-k$ and so $k^T=0$.

The $[n,1,n]$ repetition code $\mathcal{C}_R$ will be used at several points in this paper. Its parity check matrix (without redundancy) and its generating matrix can be written as
\begin{align}\label{eq:rep_code}
H_R&=\left(\begin{array}{ccccccc}
1&1&0&\dots&0&0&0\\
0&1&1&0&\dots&0&0\\
&\ddots&&\ddots&&\ddots&\\
0&0&0&\dots&0&1&1
\end{array}\right)\in\mathbb{F}_2^{(n-1)\times n},\\
G_R&=\left(\begin{array}{cccc}1&1&\dots&1\end{array}\right)\in\mathbb{F}_2^{1\times n}.
\end{align}
When we use the repetition code, its length $n$ will be context-appropriate (e.g.~so that matrix multiplications can work).

Finally, let us briefly define classical LDPC codes.
\begin{defn}[classical LDPC \cite{Gallager1962}]\label{defn:classical_ldpc}
A classical code $\mathcal{C}$ is $(b,c)$-LDPC if there is a matrix $H\in\mathbb{F}_2^{(n-k)\times n}$ such that $\ker(H)=\mathcal{C}$, every column contains at most $b$ $1$s, and every row contains at most $c$ $1$s. We call $H$ an LDPC set of parity checks.
\end{defn}
\noindent For example, the repetition code is $(2,2)$-LDPC with $H_R$ being an LDPC set of parity checks for the code.

\subsection{Quantum subsystem codes}\label{sec:subsystem}
Before diving into the description of the quantum subsystem codes in this paper (the subsequent two sections), we review in this section some of the terminology surrounding subsystem codes in general.

Quantum subsystem codes \cite{Poulin2005} are a generalization of quantum subspace codes \cite{Gottesman1997}. We restrict ourselves to the stabilizer formalism here in which both types of codes are specified by a subgroup of the Pauli group on $n$ qubits. For subspace codes, this is an abelian subgroup, the stabilizer group. For subsystem codes, this is an arbitrary subgroup, the gauge group $\mathcal{G}$. Subsystem codes are a generalization of subspace in the sense that if $\mathcal{G}$ is abelian, then the subsystem code is also a subspace code. In the general, possibly non-abelian case, we find it convenient to remove global phases from Pauli operators when defining groups of them.

Starting from the gauge group of a subsystem code, other important groups are derived.
\begin{enumerate}
\item The bare logical operators $\mathcal{L}(\mathcal{G})$: the set of all Paulis that commute with all elements of $\mathcal{G}$, also known in group theory as the centralizer of $\mathcal{G}$.
\item The stabilizers $\mathcal{S}(\mathcal{G})$: the intersection of $\mathcal{L}(\mathcal{G})$ with $\mathcal{G}$, also known as the center of $\mathcal{G}$.
\item The dressed logical operators $\hat{\mathcal{L}}(\mathcal{G})=\mathcal{G}\text{\space}\mathcal{L}(\mathcal{G})$: the centralizer of $\mathcal{S}(\mathcal{G})$.
\end{enumerate}
We point out that $\mathcal{G}=\mathcal{S}(\mathcal{G})$ if and only if the subsystem code is also a subspace code.

Code parameters are related to properties of the above groups. For instance, we denote by $K(\mathcal{G})$ the number of encoded qubits, i.e.~$4^{K(\mathcal{G})}$ is the size of $\mathcal{L}(\mathcal{G})\setminus\mathcal{S}(\mathcal{G})$, the group of logical operators modulo stabilizers. By $D(\mathcal{G})$ we denote the code distance, the weight of the lowest weight element of $\hat{\mathcal{L}}(\mathcal{G})$. 

Using a symplectic Gram-Schmidt procedure \cite{Wilde2009}, the gauge group can always be generated by
\begin{equation}
\mathcal{G}=\langle\mathcal{S}(\mathcal{G}),\overline{X}_1,\overline{Z}_1,\dots,\overline{X}_{J(\mathcal{G})},\overline{Z}_{J(\mathcal{G})}\rangle,
\end{equation}
where all generators commute except for pairs $\overline{X}_i$ and $\overline{Z}_i$. Thus, a subsystem code is seen to be a subspace code with stabilizer $\mathcal{S}(\mathcal{G})$, but including an additional $J(\mathcal{G})$ logical qubits that we do not protect. These additional logical qubits are referred to as gauge qubits. They are unprotected because error-correction proceeds by measuring a generating set of the gauge group, and thus by measuring the gauge qubits. An advantage afforded by this measurement scheme, compared to just measuring a generating set of $\mathcal{S}(\mathcal{G})$, is that the required measurements can be much lower weight. In some cases, such as the Bacon-Shor code and the subsystem codes considered in this paper, the difference in the weights of stabilizers and gauge operators can be factor of the code distance.

\subsection{Bravyi-Bacon-Shor codes}\label{sec:BBS}
Bravyi-Bacon-Shor (BBS) codes are defined entirely by a binary matrix $A\in\mathbb{F}_2^{n_1\times n_2}$. Physical qubits of the code placed on sites $(i,j)$ of a $n_1\times n_2$ square lattice $L$ for which $A_{ij}=1$. If $|A|$ is the number of 1s in $A$, there are $N=|A|$ qubits in the code. Let us take a moment to establish notation for Pauli operators on this lattice.

A Pauli $X$ or $Z$ acting on the qubit at site $(i,j)$ in the lattice is written $X_{ij}$ or $Z_{ij}$. A Pauli operator acting on multiple qubits is specified by its support.
\begin{align}
\text{For\space}S\in\mathbb{F}_2^{n_1\times n_2},&\quad X(S)=\prod_{ij}\left(X_{ij}\right)^{S_{ij}}.
\end{align}
Of course, $S$ should be such that $S_{ij}=1$ implies $A_{ij}=1$, because qubits only exist at those sites. We say $S\subseteq A$ if this is true. We also use the notation $S\cap A$ to indicate the pointwise product of binary matrices $S$ and $A$: $(S\cap A)_{ij}=S_{ij}A_{ij}$ for all $i,j$. It is always the case that $S\cap A\subseteq A$.

Conveniently, multiplication and commutation of Paulis are equivalent to addition and inner products of the support matrices,
\begin{align}\label{eq:lattice_Pauli_multiply}
X(S_1)X(S_2)&=X(S_1+S_2),\\\label{eq:lattice_Pauli_commute}
\left[X(S_1),Z(S_2)\right]&=(-1)^{\tr\left(S_1^TS_2\right)}I
\end{align}
where $[P,Q]=PQP^\dag Q^\dag$ is the group commutator and $I$ the identity operator.

From $A$ we can also define two classical codes corresponding to its column-space and row-space:
\begin{align}
\mathcal{C}_1&=\col{A},\\
\mathcal{C}_2&=\row{A}.
\end{align}
These accordingly have generating matrices $G_1$ and $G_2$, parity check matrices $H_1$ and $H_2$, and code parameters $[n_1,k,d_1]$ and $[n_2,k,d_2]$. Both $\mathcal{C}_1$ and $\mathcal{C}_2$ encode the same number of bits $k=\rank{A}=\rank{G_1}=\rank{G_2}$ because of the well-known equivalence of matrix row and column rank.

BBS codes are subsystem codes and, as such, are described by a gauge group of Pauli operators. This gauge group can be divided into $X$-type operators and $Z$-type ones, and so in this sense BBS codes are CSS subsystem codes. The gauge group is generated by $XX$ interactions between any two qubits sharing a column of lattice $L$ and $ZZ$ interactions between any two qubits sharing a row. We can write the entire gauge groups of $X$- and $Z$-type like
\begin{align}\label{eq:ggX_bbs}
\mathcal{G}^{(\text{bbs})}_X&=\left\{X(S):G_RS=0,S\subseteq A\right\},\\\label{eq:ggZ_bbs}
\mathcal{G}^{(\text{bbs})}_Z&=\left\{Z(S):SG_R^T=0,S\subseteq A\right\},
\end{align}
recalling that $G_R=(1,1,\dots,1)$ is the generating matrix of the repetition code. Therefore, $G_RS=0$ implies that columns of $S$ have even weight and $SG_R^T=0$ implies its rows have even weight.

Bare logical operators of a subsystem code commute with all its gauge operators. In the case of BBS codes, to commute with all $Z$-type gauge operators, a bare logical $X$-type operator must be supported on entire rows of the lattice. Likewise, to commute with all $X$-type gauge operators, a bare logical $Z$-type operator must be supported on entire columns. Therefore,
\begin{align}
\mathcal{L}_X^{(\text{bbs})}&=\left\{X(S\cap A):SH_R^T=0\right\},\\
\mathcal{L}_Z^{(\text{bbs})}&=\left\{Z(S\cap A):H_RS=0\right\}.
\end{align}
An example BBS code is shown in Fig.~\ref{fig:622} with the gauge operators highlighted in part (a) and the logical operators in part (b).

\begin{figure}
\centering
\includegraphics[width=0.9\columnwidth]{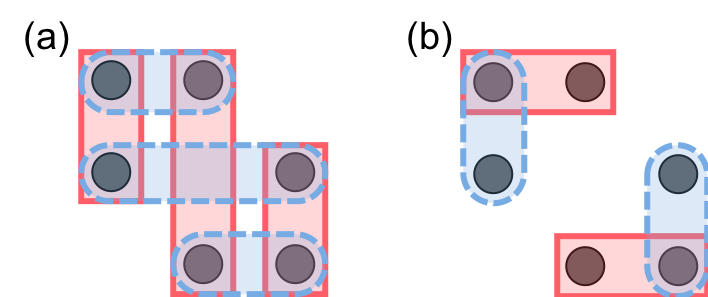}
\caption{\label{fig:622} A $\llbracket 6,2,2\rrbracket$ Bravyi-Bacon-Shor code corresponding to $A=\left(\begin{smallmatrix}1&1&0\\1&0&1\\0&1&1\end{smallmatrix}\right)$ \cite{Bravyi2011}. In (a), we encircle the supports of $X$-type (red, square, solid) and $Z$-type (blue, rounded, dashed) gauge operators. In (b), we show the supports of $X$- and $Z$-type logical operators for the two encoded qubits. We do not show them, but there are just two stabilizers, $X^{\otimes 6}$ and $Z^{\otimes 6}$.}
\end{figure}

When performing error-correction with a subsystem code, a complete generating set of gauge operators is measured. However, since not all gauge operators commute, the only reliable information gathered from this process is the eigenvalues of the stabilizers, the elements of the gauge group that do in fact commute with all gauge operators. In other words, the stabilizer is the intersection of the group of bare logical operators with the gauge group.
\begin{align}\label{eq:SbbsX}
\mathcal{S}_X^{(\text{bbs})}&=\mathcal{L}_X^{(\text{bbs})}\cap\mathcal{G}^{(\text{bbs})}_X\\
&=\left\{X(S\cap A):SH_R^T=0,G_1S=0\right\},\\\label{eq:SbbsZ}
\mathcal{S}_Z^{(\text{bbs})}&=\mathcal{L}_Z^{(\text{bbs})}\cap\mathcal{G}^{(\text{bbs})}_Z\\
&=\left\{Z(S\cap A):H_RS=0,SG_2^T=0\right\}.
\end{align}
Here $G_1S=0$ demands that each column of $S$ is a parity check of code $\mathcal{C}_1$ and thus intersects columns of $A$, which are codewords of $\mathcal{C}_1$, at an even number of places. Thus, $S\cap A$ has an even number of 1s in each column and this implies $X(S\cap A)$ is in $\mathcal{G}_X^{(\text{bbs})}$. Similar reasoning holds for the $Z$-type stabilizers.

The number of encoded qubits $K$ can be determined by counting the number of bare logical operators that are inequivalent under multiplication by stabilizers. That is, we would like the size of the quotient group $\mathcal{L}^{(\text{bbs})}_X/\mathcal{S}_X^{(\text{bbs})}$.
\begin{equation}
|\mathcal{L}^{(\text{bbs})}_X/\mathcal{S}_X^{(\text{bbs})}|=\frac{|\mathcal{L}^{(\text{bbs})}_X|}{|\mathcal{S}_X^{(\text{bbs})}|}=\frac{2^{n_1}}{|\ker(G_1)|}=2^k.
\end{equation}
Likewise, $|\mathcal{L}^{(\text{bbs})}_Z/\mathcal{S}_Z^{(\text{bbs})}|=2^k$. This implies $K=k=\rank{A}$ encoded qubits.

Dressed logical operators are bare logical operators multiplied by any number of gauge operators.
\begin{align}
\hat{\mathcal{L}}^{(\text{bbs})}_X&=\mathcal{G}^{(\text{bbs})}_X\mathcal{L}^{(\text{bbs})}_X,\\
\hat{\mathcal{L}}^{(\text{bbs})}_Z&=\mathcal{G}^{(\text{bbs})}_Z\mathcal{L}^{(\text{bbs})}_Z.
\end{align}
Equivalently, dressed logical operators are exactly those Pauli operators that commute with all stabilizers. The distance $D$ of the BBS code is the minimum nonzero weight of a dressed logical operator.

To calculate $D$, imagine first taking $X(S\cap A)\in\mathcal{L}_X^{(\text{bbs})}$ and reducing its weight by multiplying by gauge operators from $\mathcal{G}^{(\text{bbs})}_X$, which are two-qubit $X$ operators within columns. Clearly then, each column of $S\cap A$ can at best be reduced to contain either zero or one $1$ depending on the parity of the number of $1$s in that column. We calculate the parity of a column by taking its dot product with $G_R=(1,1,\dots,1)$. Thus,
\begin{equation}
\min_{g\in\mathcal{G}_X^{(\text{bbs})}}\left|gX\left(S\cap A\right)\right|=\left|G_R(S\cap A)\right|.
\end{equation}
Note that $SH_R^T=0$ if and only if rows of $S$ are codewords of the classical repetition code, i.e.~all 1s or all 0s. Accordingly, for some $\vec r\in\mathbb{F}_2^{n_1}$, $S\cap A=\diag{\vec r}A$, where $\diag{\vec r}$ is the square, diagonal matrix with $\vec r$ along the diagonal. Thus, $|G_R(S\cap A)|=|\vec rA|$, and
\begin{align}
D_X&=\min\left\{|q|>0:q\in\hat{\mathcal{L}}_X^{(\text{bbs})}\right\}\\
&=\min\left\{|\vec rA|>0:\vec r\in\mathbb{F}_2^{n_1}\right\}\\
&=\min\{|\vec x|>0:\vec x\in\row{A}\}\\
&=d_2,
\end{align}
by definition of the code distance of $\mathcal{C}_2=\row{A}$. Likewise,
\begin{align}
D_Z&=\min\left\{|q|>0:q\in\hat{\mathcal{L}}_Z^{(\text{bbs})}\right\}\\
&=\min\{|\vec x|>0:\vec x\in\col{A}\}\\
&=d_1.
\end{align}
The overall code distance of the BBS code is $D=\min(D_Z,D_X)=\min(d_1,d_2)$.

The discussion so far has reproduced Bravyi's theorem
\begin{thm}[Bravyi \cite{Bravyi2011}]\label{thm:Bravyi-Bacon-Shor}
The Bravyi-Bacon-Shor code constructed from $A\in\mathbb{F}_2^{n_1\times n_2}$, denoted $\text{BBS}(A)$, is an $\llbracket N,K,D\rrbracket$ quantum subsystem code with gauge group generated by 2-qubit operators and
\begin{align}
N&=|A|,\\
K&=\rank{A},\\ 
D&=\min\{|\vec y|>0:\vec y\in\row{A}\cup\col{A}\}.
\end{align}
\end{thm}

Assuming without loss of generality that no row or column of $A$ is all $0$s (if there is such a row or column, then it can be removed without changing the code), it is worth noting the bounds
\begin{equation}\label{eq:Aweight_bounds}
D\min(n_1,n_2)\le \min(D_Xn_1,D_Zn_2)\le |A|\le n_1n_2.
\end{equation}
The second inequality is based off the fact that each row (column) of $A$ needs to contain at least $D_X$ ($D_Z$) qubits.

\subsection{Augmented Bravyi-Bacon-Shor codes}\label{sec:aBBS}
In this subsection, we discuss geometric locality of the BBS codes. In particular, we review the modification that makes them local in 2-dimensions.

\begin{defn}[quantum LDPC codes]\label{def:quantum_LDPC}
A subsystem code with gauge group $\mathcal{G}$ is $(\beta,\gamma)$-LDPC if, there is a subset $\mathcal{G}_{\text{ldpc}}\subseteq\mathcal{G}$ such that
\begin{itemize}
\item $\mathcal{G}_{\text{ldpc}}$ generates $\mathcal{G}$, i.e.~$\mathcal{G}=\langle\mathcal{G}_{\text{ldpc}}\rangle$.
\item each qubit is in the support of at most $\beta$ of the $g\in\mathcal{G}_{\text{ldpc}}$.
\item the support of each $g\in\mathcal{G}_{\text{ldpc}}$ contains at most $\gamma$ qubits.
\end{itemize}
We refer to $\mathcal{G}_{\text{ldpc}}$ as an LDPC generating set.
\end{defn}

Every BBS code is $(4,2)$-LDPC. An LDPC generating set $\mathcal{G}_{\text{ldpc}}$ contains just the two-qubit gauge operators between consecutive qubits in a row or column.

\begin{defn}[quantum geometric locality]\label{def:quantum_geo_local}
An infinite family of $(\beta,\gamma)$-LDPC subsystem codes is local in $M$-dimensions if there is a constant $\rho$ such that all codes in the family have an LDPC generating set $\mathcal{G}_{Md}$ and the qubits of the code can be arranged on vertices of an $M$-dimensional (hyper)cubic lattice in such a way that no two qubits in the support of the same $g\in\mathcal{G}_{Md}$ are more than (Manhattan) distance $\rho$ apart.
\end{defn}

To attempt to show that a family of BBS codes is local in 2-dimensions, one might try $\mathcal{G}_{2d}=\mathcal{G}_{\text{ldpc}}$ from above. While this is of course an LDPC generating set, it is not necessarily true that elements of $\mathcal{G}_{2d}$ are supported in constant-sized regions of the 2-dimensional lattice. The difficulty is that $A$ may contain two consecutive $1$s in the same row or column that are separated by many $0$s (potentially a number of $0$s that grows with code size) and thus consecutive qubits are far apart.

To remedy this, Bravyi \cite{Bravyi2011} introduces two more qubits at every site $(i,j)$ such that $A_{ij}=0$. One qubit participates in the two-qubit gauge operators of row $i$ and the other in the gauge operators of column $j$. Hence, we now say that there are three types of qubits making up the code -- type 0 qubits reside at sites where $A_{ij}=1$, whereas type 1 and type 2 qubits reside at sites where $A_{ij}=0$. These qubit types can be used to define two lattices -- $L_1$ consists of qubits of type 0 and type 1 and $L_2$ consists of qubits of type 0 and type 2. It is important to note that the lattices share the type 0 qubits, i.e.~the lattices are identified at the sites where $A_{ij}=1$.

To distinguish Paulis acting on qubits in lattices $L_1$ or $L_2$, we use superscripts, e.q.~$X^{(L_1)}_{ij}$ or $X^{(L_2)}_{ij}$ for single-qubit Paulis and $X^{(L_1)}(S)$ or $X^{(L_2)}(S)$ for Paulis acting on multiple qubits specified by support $S$. Of course, due to the identification of qubits between $L_1$ and $L_2$, a particular (say, $X$-type) Pauli $P$ does not have unique supports $S_1,S_2$ such that $P=X^{(L_1)}(S_1)X^{(L_2)}(S_2)$. Indeed, letting $\mathbbm{1}$ be the matrix of all 1s,
\begin{equation}
X^{(L_1)}(S_1)X^{(L_2)}(S_2)=X^{(L_1)}(T_1)X^{(L_2)}(T_2)
\end{equation}
if and only if $S_1\cap(\mathbbm{1}-A)=T_1\cap(\mathbbm{1}-A)$, $S_2\cap(\mathbbm{1}-A)=T_2\cap(\mathbbm{1}-A)$, and $(S_1\cap A)+(S_2\cap A)=(T_1\cap A)+(T_2\cap A)$.

Using this notation, the augmented Bravyi-Bacon-Shor code (aBBS) has gauge groups
\begin{align}\label{eq:ggX_aBBS}
\mathcal{G}^{(\text{abbs})}_X&=\{X^{(L_1)}(S)X^{(L_2)}(T):G_RS=0,T\subseteq\mathbbm{1}-A\},\\\label{eq:ggZ_aBBS}
\mathcal{G}^{(\text{abbs})}_Z&=\{Z^{(L_1)}(S)Z^{(L_2)}(T):TG_R^T=0,S\subseteq\mathbbm{1}-A\}.
\end{align}

Intuition for this gauge group arises by developing a generating set local in 2-dimensions. This generating set $\mathcal{G}_{2d}$ can be chosen to be the set of all two-qubit gauge operators on \emph{neighboring} qubits in the lattices as well as all one-qubit gauge operators. That is, with $[t]=\{1,2,\dots,t\}$, we have
\begin{align}\label{eq:g2d}
\mathcal{G}_{2d}=&\{X^{(L_1)}_{ij}X^{(L_1)}_{i+1,j}:i\in[n_1-1],j\in[n_2]\}\\
&\cup\{Z^{(L_2)}_{ij}Z^{(L_2)}_{i,j+1}:i\in[n_1],j\in[n_2-1]\}\\
&\cup\{X^{(L_2)}_{ij}:A_{ij}=0\}\\
&\cup\{Z^{(L_1)}_{ij}:A_{ij}=0\}.
\end{align}
This set, which has $4n_1n_2-(n_1+n_2)-2|A|$ independent generators, is clearly local in 2-dimensions: it consists of two-qubit $X$ operators between qubits sharing a column in lattice 1, two-qubit $Z$ operators between qubits sharing a row in lattice 2, and single-qubit operators on type 1 and type 2 qubits. An example of $\mathcal{G}_{2d}$ for a $\llbracket 12,2,2\rrbracket$ aBBS code is shown in Fig.~\ref{fig:aug_622}.

\begin{figure}
\centering
\includegraphics[width=0.9\columnwidth]{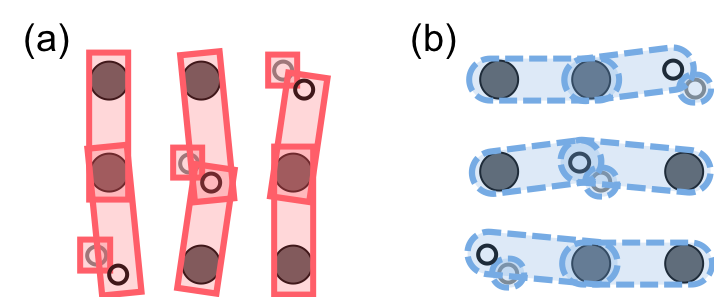}
\caption{\label{fig:aug_622} Generating sets of (a) $X$-type and (b) $Z$-type gauge operators for the augmented version of the code from Fig.~\ref{fig:622}. Type 0 qubits are shown as large, filled circles, while type 1 and 2 qubits are small and unfilled. The generating sets consist entirely of two-qubit (dark) and single-qubit (light) operators.}
\end{figure}

Bare logical operators and stabilizers of an aBBS code are derived similarly to those of a BBS code. Rather than go through those arguments again, we just record the results here.
\begin{align}
\mathcal{L}^{(\text{abbs})}_X&=\{X^{(L_2)}(S):SH_R^T=0\},\\
\mathcal{L}^{(\text{abbs})}_Z&=\{Z^{(L_1)}(S):H_RS=0\},\\\label{eq:SabbsX}
\mathcal{S}^{(\text{abbs})}_X&=\{X^{(L_2)}(S):SH_R^T=0,G_1S=0\},\\\label{eq:SabbsZ}
\mathcal{S}^{(\text{abbs})}_Z&=\{Z^{(L_1)}(S):H_RS=0,SG_2^T=0\}.
\end{align}
Code parameters $K$ and $D$ are also unchanged. Collecting this into a theorem, we have:
\begin{thm}[Bravyi \cite{Bravyi2011}]\label{thm:augmented_Bravyi-Bacon-Shor}
The augmented Bravyi-Bacon-Shor code constructed from $A\in\mathbb{F}_2^{n_1\times n_2}$, denoted $\text{aBBS}(A)$, is an $\llbracket N,K,D\rrbracket$ quantum subsystem code that is local in 2-dimensions, has a gauge group generated by 1- or 2-qubit operators, and
\begin{align}
N&=2n_1n_2-|A|,\\
K&=\rank{A},\\
D&=\min\{|\vec y|>0:\vec y\in\row{A}\cup\col{A}\}.
\end{align}
\end{thm}

\subsection{Hypergraph Product Codes}\label{sec:HGP}
Introduced in \cite{Tillich2014}, the hypergraph product takes two classical parity check matrices $H_1\in\mathbb{F}_2^{n_1^T\times n_1}$ and $H_2\in\mathbb{F}_2^{n_2^T\times n_2}$ and produces a quantum code. The code ultimately is of CSS type (though the traditional CSS construction is not used to obtain it) and so its stabilizer group can be separated into stabilizers of Pauli $X$-type and those of Pauli $Z$-type.

Our description of the hypergraph product is a little unconventional but is in line with how we described BBS and aBBS codes, making it easier to relate the two later. It is essentially a description in terms of the ``reshaped" matrices used at some points by Campbell \cite{Campbell2018}.

In our notation, qubits of the hypergraph product code are placed on the vertices of two square lattices (see Fig.~\ref{fig:hypergraph_prod}). The first lattice $L$ is $n_1\times n_2$. The second lattice $l$ is $n_1^T\times n_2^T$. A Pauli $X$ or $Z$ acting on the qubit at site $(i,j)$ in lattice $L$ is denoted $X_{ij}^{(L)}$ or $Z_{ij}^{(L)}$ and similarly for Paulis acting in lattice $l$. A Pauli operator acting on multiple qubits is specified by its support, e.g.~$X^{(L)}(S)$ or $Z^{(L)}(S)$, just as for the BBS and aBBS codes.

\begin{figure}
\includegraphics[width=0.475\columnwidth]{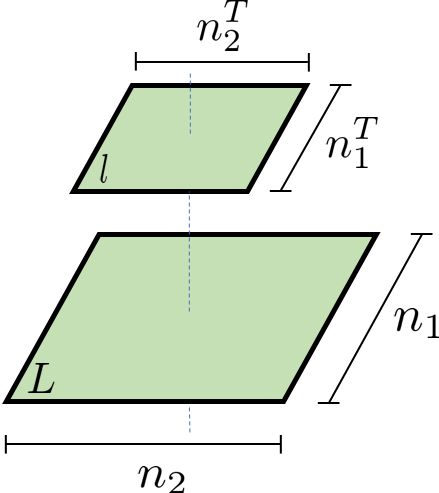}
\caption{\label{fig:hypergraph_prod} The two lattices of qubits that make up a hypergraph product code.}
\end{figure}

On the classical side, we define generating matrices $G_1$ and $G_2$ for the classical codes $\mathcal{C}_1$ and $\mathcal{C}_2$ corresponding to $H_1$ and $H_2$. We define their code parameters as $[n_1,k_1,d_1]$ and $[n_2,k_2,d_2]$ and assume without loss of generality that they are nontrivial: $k_1,k_2>0$. Similarly, let $F_1$ and $F_2$ be generating matrices for codes $\mathcal{C}_1^T$ and $\mathcal{C}_2^T$ with code parameters $[n_1^T,k_1^T,d_1^T]$ and $[n_2^T,k_2^T,d_2^T]$. Either transpose code may be trivial, in which case we define its distance to be infinite.

Using this notation, the hypergraph product of $H_1$ and $H_2$ is a quantum code $\text{HGP}(H_1,H_2)$ defined by the following sets of stabilizers, divided into $X$-type and $Z$-type,
\begin{alignat}{3}\label{eq:ShgpX}
\mathcal{S}^{(\text{hgp})}_X&=\big\{X^{(L)}(S)X^{(l)}(T):&SH_2^T&=H_1^TT,\\\nonumber
&&G_1S&=0,TF_2^T=0\big\},\\\label{eq:ShgpZ}
\mathcal{S}^{(\text{hgp})}_Z&=\big\{Z^{(L)}(S)Z^{(l)}(T):&H_1S&=TH_2,\\\nonumber
&&SG_2^T&=0,F_1T=0\big\}.
\end{alignat}

To show these stabilizers commute, let $M=X^{(L)}(S)X^{(l)}(T)\in\mathcal{S}^{(\text{hgp})}_X$ and $M'=Z^{(L)}(S')Z^{(l)}(T')\in\mathcal{S}^{(\text{hgp})}_Z$. By Eq.~\eqref{eq:lattice_Pauli_commute}, we need to show $\tr(S^TS')+\tr(T^TT')=0$. Notice that $G_1S=0$ demands that columns of $S$ are parity checks for $\mathcal{C}_1$. In other words, there exists $A$ such that $S=H_1^TA$. Likewise, because $TF_2^T=0$, rows of $T$ are parity checks for $\mathcal{C}_2^T$, or, equivalently, there exists $B$ such that $T=BH_2^T$. Finally, the same reasoning holds for $S'$ and $T'$, showing the existence of $A'$ and $B'$ such that $S'=A'H_2$ and $T'=H_1B'$. Since $SH_2^T=H_1^TT$ and $H_1S'=T'H_2$, we have $H_2A^TH_1=H_2B^TH_1$ and $H_1A'H_2=H_1B'H_2$. Putting it all together we have
\begin{align}
\tr(S^TS')&=\tr(A^TH_1A'H_2)\\
&=\tr(B^TH_1B'H_2)=\tr(T^TT'),
\end{align}
completing the proof.

In Appendix~\ref{app:hypergraph_product_codes}, we connect this description of the hypergraph product code with the original definition, and derive other relevant properties. We note here that logical operators for $\text{HGP}(H_1,H_2)$ are
\begin{align}\label{eq:LhgpX}
\mathcal{L}_X^{(\text{hgp})}&=\{X^{(L)}(S)X^{(l)}(T):SH_2^T=H_1^TT\},\\\label{eq:LhgpZ}
\mathcal{L}_Z^{(\text{hgp})}&=\{Z^{(L)}(S)Z^{(l)}(T):H_1S=TH_2\},
\end{align}
and it has code parameters $\llbracket N,K,D\rrbracket$ \cite{Tillich2014}:
\begin{align}\label{eq:hgpN}
N&=n_1n_2+n_1^Tn_2^T,\\\label{eq:hgpK}
K&=k_1k_2+k_1^Tk_2^T,\\\label{eq:hgpD}
D&=\bigg\{\begin{array}{lr}\min(d_1,d_2),&k_1^T=0\text{ or }k_2^T=0\\\min(d_1,d_2,d_1^T,d_2^T),&\text{otherwise}\end{array}.
\end{align}
Lastly, if $H_1$ is a $(b_1,c_1)$-LDPC set of parity checks and $H_2$ is a $(b_2,c_2)$-LDPC set of parity checks, then $\text{HGP}(H_1,H_2)$ is $(\beta,\gamma)$-LDPC for \cite{Tillich2014}
\begin{align}\label{eq:hgpBeta}
\beta&=\max(b_1+b_2,c_1+c_2),\\\label{eq:hgpGamma}
\gamma&=\max(c_1+b_2,b_1+c_2).
\end{align}

\section{Constructing and decoding optimal 2-dimensional subsystem codes}\label{sec:constructing_and_decoding}
In this section, we show how to make BBS codes from two classical codes. We note that using good classical codes leads to optimal scaling of the quantum code parameters and show how classical decoders are used to decode the quantum codes.

\subsection{Bravyi-Bacon-Shor codes from classical codes}
In Section~\ref{sec:BBS} we noted that an $\llbracket N,K,D\rrbracket$ BBS code specified by matrix $A\in\mathbb{F}_2^{n_1\times n_2}$ defines two classical codes $\mathcal{C}_1=\col{A}$ and $\mathcal{C}_2=\row{A}$, and that, if those classical codes have parameters $[n_1,k,d_1]$ and $[n_2,k,d_2]$, we have code parameter relations $K=k$ and $D=\min(d_1,d_2)$. The goal now is to explore the converse: given two classical codes $\mathcal{C}_1$ and $\mathcal{C}_2$, how should we construct a BBS code with the same relations in code parameters?

Suppose that the classical codes have generating matrices $G_1\in\mathbb{F}_2^{k\times n_1}$ and $G_2\in\mathbb{F}_2^{k\times n_2}$. We then construct the code $\text{BBS}(A)$ with
\begin{equation}
A=G_1^TQG_2\in\mathbb{F}_2^{n_1\times n_2},
\end{equation}
where $Q$ is any full-rank $k\times k$ matrix representing the non-uniqueness of the generating matrices. Adjusting $Q$ can change the number of physical qubits in the code.

Now notice that
\begin{align}
\col{A}&=\{G_1^TQG_2\vec x:\vec x\in\mathbb{F}_2^{n_2}\}\\
&=\{G_1^TQ\vec y:\vec y\in\mathbb{F}_2^k\}\\
&=\{G_1^T\vec z:\vec z\in\mathbb{F}_2^k\}\\
&=\col{G_1^T}=\row{G_1}=\mathcal{C}_1.
\end{align}
The second equality relies on $G_2$ being full-rank and the third on $Q$ being full-rank. Likewise, similar reasoning shows that $\row{A}=\mathcal{C}_2$.

Therefore, we have the following theorem.
\begin{thm}\label{thm:BBS_from_classical}
For all full-rank $Q\in\mathbb{F}_2^{k\times k}$ and every two classical codes $\mathcal{C}_1$, $\mathcal{C}_2$ with parameters $[n_1,k,d_1]$, $[n_2,k,d_2]$ and generating matrices $G_1\in\mathbb{F}_2^{k\times n_1}$, $G_2\in\mathbb{F}_2^{k\times n_2}$, let $A=G_1^TQG_2$. Then $\text{BBS}(A)$ is an $\llbracket N,K,D\rrbracket$ quantum subsystem code and $\text{aBBS}(A)$ an $\llbracket N_{2d},K,D\rrbracket$ subsystem code local in 2-dimensions with
\begin{align}
\min(n_1d_2,d_1n_2)\le N&\le n_1n_2,\\
n_1n_2\le N_{2d}&\le 2n_1n_2-\min(n_1d_2,d_1n_2),\\
K&=k,\\
D&=\min(d_1,d_2).
\end{align}
\end{thm}
\noindent The lower bound on $N$ and upper bound on $N_{2d}$ are provided by Eq.~\eqref{eq:Aweight_bounds}.

Before discussing the theorem's implications, let us briefly present some examples, starting with the Bacon-Shor code.
\begin{example}\label{ex:bacon_shor}
Let $\mathcal{C}_1=\mathcal{C}_2=\mathcal{C}_R$ be the $[n,1,n]$ repetition code (see Eq.~\eqref{eq:rep_code}). Then $A=G_R^TG_R=\mathbbm{1}$ (the all $1$s matrix) represents a Bravyi-Bacon-Shor with a qubit at every lattice site, $X$-type ($Z$-type) gauge operators between pairs of qubits in the same column (row), and $X$-type ($Z$-type) stabilizers that span pairs of rows (columns). That is, we have reconstructed the Bacon-Shor code \cite{Bacon2006}.
\end{example}

\begin{example}
The $[7,4,3]$ Hamming code is generated by
\begin{equation}\nonumber
G=\left(\begin{array}{ccccccc}
1&0&0&0&1&1&0\\
0&1&0&0&1&0&1\\
0&0&1&0&0&1&1\\
0&0&0&1&1&1&1
\end{array}\right)\text{\hspace{-2pt}},\text{\hspace{1pt}}
H=\left(\begin{array}{ccccccc}
1&1&0&1&1&0&0\\
1&0&1&1&0&1&0\\
0&1&1&1&0&0&1
\end{array}\right).
\end{equation}
Let $A=G^TQG$ for full-rank $4\times4$ matrix $Q$. Taking $Q=I$ gives a $\llbracket25,4,3\rrbracket$ Bravyi-Bacon-Shor code:
\begin{equation}\nonumber
A=\left(\begin{array}{ccccccc}
1&0&0&0&1&1&0\\
0&1&0&0&1&0&1\\
0&0&1&0&0&1&1\\
0&0&0&1&1&1&1\\
1&1&0&1&1&0&0\\
1&0&1&1&0&1&0\\
0&1&1&1&0&0&1
\end{array}\right).
\end{equation}
Alternatively, taking $Q=\left(\begin{smallmatrix}0&0&1&0\\0&1&0&1\\1&0&0&0\\0&1&0&0\end{smallmatrix}\right)$ minimizes the number of qubits, giving a $\llbracket21,4,3\rrbracket$ Bravyi-Bacon-Shor code. 
\end{example}

In \cite{Bravyi2011}, Bravyi shows that for any family of $\llbracket N,K,D\rrbracket$ quantum subsystem codes local in 2-dimensions $KD=O(N)$. He then provides a nonconstructive argument that families of aBBS codes exist that saturate this bound. Theorem~\ref{thm:BBS_from_classical} elucidates this existence proof by connecting it to the classical case. If we have a family of good $[n,k,d]$ classical codes -- i.e.~there are constants $\alpha,\beta$ such that for all $n$, $k\ge\alpha n$ and $d\ge\beta n$ -- then the aBBS code family created from Theorem~\ref{thm:BBS_from_classical} satisfies $KD=\alpha\beta n^2\ge\alpha\beta N_{2d}/2$. So the aBBS codes created this way saturate Bravyi's bound.

Moreover, Theorem~\ref{thm:BBS_from_classical} provides the means to elevate the aforementioned nonconstructive proof to an explicit constructive proof. One only needs an explicit construction of good classical codes. Such constructions exist, e.g.~expander codes. We review these classical codes in detail in Appendix~\ref{app:expander_codes_construction}.

Finally, we should point out that although Theorem~\ref{thm:BBS_from_classical} produces BBS and aBBS codes for which $K,D=O(\sqrt{N})$, it can be used to trade off $K$ and $D$. Bravyi and Terhal \cite{Bravyi2009} have shown that $D=O(\sqrt{N})$ for subsystem codes in 2-dimensions. So, assume that we would like a code family with $K=\alpha N^{1-a}$ and $D=\beta N^a$ for some constants $a\le1/2$, $\alpha$, and $\beta$. To construct this code family, use Theorem~\ref{thm:BBS_from_classical} and a good family of classical codes to make a quantum code family with $K=\alpha N^a$, $D=\beta N^a$, and $O(N^{2a})$ physical qubits. Take $N^{1-2a}$ copies of this family to make the desired family with $K=\alpha N^{1-2a}N^a=\alpha N^{1-a}$, $D=\beta N^a$, and $N^{1-2a}O(N^{2a})=O(N)$ physical qubits.

If, for whatever reason, a family with parameters $KD=o(N)$ is desired (i.e.~$KD$ scales strictly less than $N$), then one can take a family with $K'D=\Theta(N)$ and ignore a fraction $1-K/K'$ of the encoded qubits. See Fig.~\ref{fig:code_region} for a summary of the last two paragraphs.

\begin{figure}
\centering
\includegraphics[width=0.75\columnwidth]{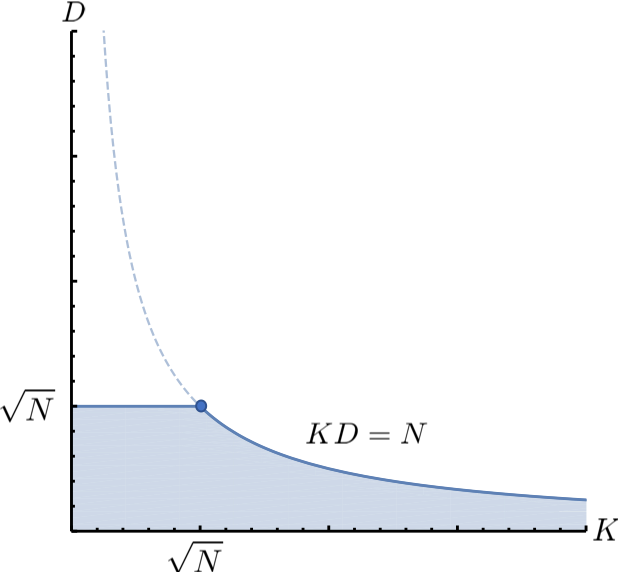}
\caption{\label{fig:code_region} The region of possible \cite{Bravyi2009,Bravyi2011} 2-dimensional subsystem $\llbracket N,K,D\rrbracket$ code families. We can construct families at all these points by appealing to constructions of good classical codes and using Theorem~\ref{thm:BBS_from_classical}.}
\end{figure}

\subsection{Decoding BBS codes}
Correcting errors on a quantum subsystem code involves (1) measuring a generating set of gauge operators, (2) reconstructing the values of the stabilizers from the results, and (3) applying a Pauli correction. An advantage of the BBS or aBBS codes is that the generating set of gauge operators includes only 2-qubit operators (see Eq.~\ref{eq:g2d}) despite the stabilizers being high-weight. In this section, we identify another convenient feature -- the correction in the third step can be calculated by decoders for the corresponding classical codes.

Let us begin by making some assumptions about the error model and codes. While not essential to the main conclusions, these assumptions simplify the discussion. Regarding the error model, we assume that each qubit suffers an $X$ error with probability $q$ and, independently, a $Z$ error with the same probability $q$. Two-qubit Pauli measurements (e.g.~of the gauge operators) are assumed to fail with probability $q'$. Regarding the codes, we assume that $A$ is an $n\times n$ symmetric matrix, so there is just one $[n,k,d]$ classical code $\mathcal{C}=\row{A}=\col{A}$ under consideration. We let $H$ and $G$ be the parity check and generating matrices of this code.

What is essential to our conclusions here is the existence of a decoding algorithm $\mathcal{D}$ for the classical code $\mathcal{C}$ of the following form. Decoder $\mathcal{D}$ takes as input faulty parity check information from a faulty codeword $\vec i=H(\vec w+\vec e)+\vec f$, where $\vec w\in\mathcal{C}$, $\vec e$ represents data errors, and $\vec f$ represents errors in the ``measurement" of the parity checks (though a more appropriate classical terminology might be the ``calculation" of the checks). The decoder's task is then to find a recovery $\pvec e'=\mathcal{D}(\vec i)$ that is close to $\vec e$, and update the classical state from $\vec w+\vec e$ to $\vec w+\vec e+\pvec e'$. This process is repeated some number of rounds, alternating the application of random noise $\vec e$ and $\vec f$ with decoding. Afterwards, we imagine ``ideal" decoding with $\vec f=\vec 0$ is performed, and if $\vec w$ is not the final state, then the error-correction has failed. Failure (after the given number of rounds) occurs with some probability $\bar p$, which is a function of the probability distribution of errors $\vec e$ and $\vec f$.

Generally, in classical error-correction, measurement of the parity checks is considered perfect, and so decoders with the required capability of dealing with measurement errors are seldom created or used. However, expander codes do have a suitable decoder, the flip decoder, which is discussed in Appendix~\ref{app:expander_codes_decoding}.

We discuss the decoding of (symmetric, $A=A^T$) BBS codes in detail, then in the next section briefly discuss the case of (symmetric) aBBS codes, which is similar. Our reasoning hinges on associating the stabilizers of the BBS code with parity checks of the classical code $\mathcal{C}$ and the dressed logical operators of the BBS code with the codewords of $\mathcal{C}$.

To realize these associations, we rewrite $\mathcal{S}_X^{(\text{bbs})}$ from Eq.~\eqref{eq:SbbsX}. Let $X(S\cap A)\in\mathcal{S}_X^{(\text{bbs})}$. Since $SH_R^T=0$, rows of $S$ are codewords of $\mathcal{C}_R$, either all 1s or all 0s. Because $GS=0$, columns of $S$ are parity checks of $\mathcal{C}$. Therefore, $S\cap A=\diag{\vec r}A$ for some $\vec r\in\row{H}$. We have
\begin{equation}\label{eq:rewrite_SbbsX}
\mathcal{S}^{(\text{bbs})}_X=\{X\left(\diag{\vec r}A\right):\vec r\in\row{H}\}.
\end{equation}
Similarly, rewrite $\mathcal{S}_Z^{\text{(bbs)}}$ as
\begin{equation}\label{eq:rewrite_SbbsZ}
\mathcal{S}^{(\text{bbs})}_Z=\{Z\left(A\hspace{2pt}\diag{\vec c}\right):\vec c\in\row{H}\}.
\end{equation}
Thus, the parity checks of the classical code indicate which sets of rows or columns constitute a stabilizer.

Dressed logical operators $\hat{\mathcal{L}}_X^{(\text{bbs})}$ are exactly those $X$-type Paulis that commute with all the $Z$-type stabilizers. But because $Z$-type stabilizers are supported on entire columns of $A$, they are only sensitive to whether an even or odd number of Pauli $X$ errors occurred within a column. Indeed, single qubit $X$ errors within a column are equivalent up to gauge operators. Say that a column is odd if it contains an odd number of $X$ errors. An $X$-type operator commutes with all the $Z$-type stabilizers if and only if it consists of odd columns corresponding to a codeword $\vec w\in\mathcal{C}$, i.e.~column $i$ is odd if and only if $\vec w_i=1$. In other words, the even or oddness of a column corresponds to the 0 or 1 state of an effective classical bit of the code $\mathcal{C}$. Symmetry of $A$ dictates that the same correspondence holds for $Z$-type dressed logical operators $\hat{\mathcal{L}}_Z^{(\text{bbs})}$ and $Z$ errors in rows.

The upshot of the previous paragraph is that to decode a $\text{BBS}(A)$ code, we may collect $X$- or $Z$-type stabilizer information $\vec\sigma$, run the classical decoder $\pvec e'=\mathcal{D}(\vec\sigma)$, and apply a $Z$- or $X$-type Pauli correction to a single qubit in each row or column indicated by $\pvec e'$. We call this the decoder induced by $\mathcal{D}$, or simply the \emph{induced decoder} for $\text{BBS}(A)$. To evaluate how well the induced decoder works, we just need to map the quantum errors to the effective classical errors that the decoder $\mathcal{D}$ sees.
 
The probability that an odd number of $X$ errors occurs within column $i$ containing $c_i$ qubits is
\begin{equation}\label{eq:mapped_bit_errors}
p_i=\sum_{\substack{l=1\\l\text{ odd}}}^{c_i}\binom{c_i}{l}q^l(1-q)^{r_i-l}=\frac12\left(1-(1-2q)^{c_i}\right).
\end{equation}
By symmetry, this situation is the same for $Z$ errors in the rows. So $p_i$ is the probability that bit $i$ has flipped in the classical code.

Similarly, stabilizers of the Bravyi-Bacon-Shor code are the product of several two-qubit gauge operators. For instance, there is an $Z$-type stabilizer $Z(A\hspace{2pt}\diag{\vec h_j})$ for row $\vec h_j$ of $H$, and it is made of $c_j'=|A\hspace{2pt}\diag{\vec h_j}|/2\le n|\vec h_j|/2$ two-qubit gauge measurements. The probability this stabilizer measurement is incorrect depends only on whether an even or odd number of its constituent gauge measurements are incorrect:
\begin{equation}\label{eq:mapped_meas_errors}
p_j'=\sum_{\substack{l=1\\l\text{ odd}}}^{c_j'}\binom{c_j'}{l}q^l(1-q')^{c'_j-l}=\frac12\left(1-(1-2q')^{c'_j}\right).
\end{equation}
By symmetry, this situation is identical for the $X$-type stabilizers. Thus, $p_j'$ is the probability that the parity check calculation for parity check $j$ is incorrect.

These relations between quantum and classical errors give us the following lemma.
\begin{lem}\label{lem:induced_decoder}
Say that using decoder $\mathcal{D}$ on the classical error model in which data errors have probabilities $p_i$ and parity check errors have probabilities $p_j'$ results in a logical error rate of $\bar p(p_i,p_j')$. The induced decoder with respect to $\mathcal{D}$ on an error model in which qubits fail with independent $X$ or $Z$ errors with probability $q$ and two-qubit Pauli measurements fail with probability $q'$ has a logical error rate
\begin{equation}\label{eq:induced_error}
\bar q(q,q')\le 2\bar p(p_i,p_j')
\end{equation}
where $p_i$ and $p_j'$ are given by Eqs.~\eqref{eq:mapped_bit_errors} and \eqref{eq:mapped_meas_errors}.
\end{lem}

The factor of two in Eq.~\eqref{eq:induced_error} results from the $X$ and $Z$ errors being decoded separately. Independent $X$, $Z$ noise is of course not critical to the lemma. For depolarizing noise for example, in which Pauli $X$, $Y$, or $Z$ errors occur with equal probability $q/3$, the logical error rate is at most $\bar q(2q/3,q')$ since $2q/3$ is the probability of a $Z$ or $X$ error. On the other hand, the induced decoder does discount the correlations in $X$ and $Z$ noise, so is not expected to be optimal in this case.

Also crucial to note is that for small, constant $q$ and $q'$, $p_i\approx c_iq$ and $p_j'\approx c_j'q'$. Because $c_i,c_j'\ge d$, the effective classical error rates increase at least proportionally to the code distance. In the limit of large code size and distance, no classical code can be expected to correct such noise, and thus this shows the lack of asymptotic threshold for BBS codes. Nevertheless, the lemma indicates a close connection between the quantum and classical error rates. If a classical code has a ``useful" (e.g.~order $10^{-a}$ for some moderately large $a$) logical error rate for $p_i<p$ and $p_j'<p'$, then the quantum code has a useful (i.e.~order $10^{-a}$) logical error rate for $q<p/(\max_ic_i)$ and $q'<p'/(\max_jc_j')$.

Lemma~\ref{lem:induced_decoder} indicates two ways to improve the decoding of BBS codes, even before tailoring to the noise. The first, more obvious way, is to find better decoders for the constituent classical codes. This is of course subject to the constraint that these classical decoders can tolerate measurement noise, which we noted previously is nonstandard but attainable for expander codes for example.

The second way to improve decoding is by reducing the values of $c_i$ (the number of qubits in row or column $i$) and $c_j'$ (the number of gauge-operators making up stabilizer $j$). This correlates roughly with minimizing $|A|$, the number of qubits in the BBS code, which can be done without change in the code parameters by appropriate choice of $Q$ in Theorem~\ref{thm:BBS_from_classical}.

Finally, let us discuss the time complexity of an induced decoder. This can be broken down into two parts: (1) the time it takes to acquire the stabilizer values that are input to $\mathcal{D}$ and (2) the time it takes to run $\mathcal{D}$ twice, once for $X$-stabilizers, once for $Z$. A particular stabilizer corresponding to a weight-$w$ parity check is the sum of $O(wn)$ two-qubit measurements and therefore takes $O(wn)$ time to compute. If $m$ stabilizer values are needed as input to the classical decoders, and the classical decoders run in time at most $t$, then induced decoding takes time $O(mwn+t)$. Using BBS codes constructed from classical expander codes as an example, the flip decoder $\mathcal{D}$ (see Appendix~\ref{app:expander_codes_decoding}) requires just $m=O(n)$ bits of input from weight $w=O(1)$ checks and runs in time $t=O(n)$. Thus, induced decoding takes time $O(n^2+n)=O(N)$, i.e.~linear in the size of the quantum code.

\subsection{Decoding aBBS codes}
In this subsection, we briefly discuss the decoding of (symmetric, $A=A^T$) aBBS codes assuming we can only measure operators in $\mathcal{G}_{2d}$, Eq.~\eqref{eq:g2d}, i.e.~two-qubit operators on neighboring qubits and some single-qubit measurements. We still advocate using the induced decoder of the previous section, but it is now more difficult to collect the stabilizer values from this restricted set of gauge operator measurements.

Similar to how we derived Eqs.~\eqref{eq:rewrite_SbbsX}, \eqref{eq:rewrite_SbbsZ}, we can rewrite the stabilizers of the aBBS codes to correspond to classical parity checks (recall, $\mathbbm{1}$ is the matrix of all 1s):
\begin{align}
\mathcal{S}^{(\text{aBBS})}_X&=\big\{X^{(L_2)}\left(\diag{\vec r}\mathbbm{1}\right):\vec r\in\row{H}\big\},\\
\mathcal{S}^{(\text{aBBS})}_Z&=\big\{Z^{(L_1)}\left(\mathbbm{1}\hspace{2pt}\diag{\vec c}\right):\vec c\in\row{H}\big\}.
\end{align}
This leads to similar conclusions about errors on the effective bits of the classical code $\mathcal{C}$. With the recognition that $c_i=n$ for all $i$, Eq.~\eqref{eq:mapped_bit_errors} still represents the probability of error for an effective classical bit.

As one may expect, because we have restricted what gauge operators may be measured to those in $\mathcal{G}_{2d}$, aBBS decoding also differs from BBS decoding in how eigenvalues of the stabilizers are calculated. If $\vec h_j$ is a row of $H$ and $S=Z^{(L_1)}(\mathbbm{1}\hspace{2pt}\diag{\vec h_j})$ is the corresponding stabilizer, then we should let $c_j'$ be the minimal number of elements of $\mathcal{G}_{2d}$ whose product is $S$. Since $S$ may include rows that are $O(n)$ distance apart, $c_j'$ may be as a large as $O(n^2)$. With this redefinition of $c_j'$ however, Eq.~\eqref{eq:mapped_meas_errors} again represents the probability of error for a parity check. Lemma~\ref{lem:induced_decoder} holds given these changes to $c_i$ and $c_j'$.

Now we discuss the runtime. Because $c_j'$ can be so large, we may be worried that it takes more time to decode, because ostensibly stabilizers corresponding to even just constant-weight parity checks may be the sum of as many as $O(n^2)$ elements of $\mathcal{G}_{2d}$ (and note that $|\mathcal{G}_{2d}|=O(n^2)$). However, a simple application of dynamic programming solves this. Suppose that we measure all two-body $Z$-gauge operators and get values $m_{ij}\in\{0,1\}$ corresponding to positions $(i,j)$ in the lattice. We can sweep across the lattice calculating the cumulative values across rows
\begin{equation}
M_{ij}=\sum_{l=1}^jm_{il}
\end{equation}
using just $O(n^2)=O(N)$ time. $Z$-type stabilizers corresponding to constant-weight parity checks are once again the sum of $O(n)$ of the $M_{ij}$ as well as $O(n)$ single-qubit measurements. Symmetry dictates the same is true for $X$-type stabilizers. Therefore, for example, the induced decoder with respect to the flip decoder for aBBS codes constructed from classical expander codes can still be implemented in linear time.

\section{Gauge-fixing}\label{sec:gauge_fixing}
In this section, we show that an aBBS code can be gauge-fixed to the corresponding BBS code and to certain hypergraph product codes. We begin, however, by defining gauge-fixing in general.

\subsection{Definition}
As we discussed in Section~\ref{sec:subsystem}, one way to think about subsystem codes is that in addition to the logical qubits encoded in the code, there are additional encoded qubits, the gauge qubits, which we do not care about protecting. In fact, the logical operators for these gauge qubits may be very low weight -- they are the gauge operators that we measure to perform error-correction.

The existence of gauge qubits, however, leads us to imagine a family of related codes in which some or all of the gauge qubits are fixed to some stabilizer state $\ket{\psi_g}$. In these related codes, called gauge-fixings, we have removed some or all of the gauge degrees of freedom by removing operators from the gauge group that do not stabilize $\ket{\psi_g}$. Generally, this makes error-correction more difficult -- a generating set for the new gauge group may necessarily contain higher weight operators -- but by reducing the size of the group of dressed logical operators, the environment has fewer ways to introduce logical errors to the data. This may even result in asymptotic error-correction thresholds in the gauge-fixed codes where none existed in the original subsystem code. A well-known example is the gauge-fixing of the Bacon-Shor code to the surface code \cite{Li2018}.



To discuss gauge-fixing in general, we use the following definition, using the notation from Section~\ref{sec:subsystem}.
\begin{defn}\label{defn:gauge_fixing}
We say that $\mathcal{G}'$ is a gauge-fixing of $\mathcal{G}$ if
\begin{enumerate}
\item $\mathcal{S}(\mathcal{G})\le\mathcal{S}(\mathcal{G}')\le\mathcal{G}'\le\mathcal{G}$
\item $K(\mathcal{G})=K(\mathcal{G}')$
\end{enumerate}
Generalizing the language slightly, we also say that a code $\mathcal{Q}'$ is a gauge-fixing of a code $\mathcal{Q}$ if their gauge groups are related appropriately.
\end{defn}

By the definition, a subsystem code and its gauge-fixing have the same total number of physical qubits and logical qubits. We can also say something about their code distances. 
\begin{lem}\label{lem:gauge_fixing_distance}
If $\mathcal{G}'$ is a gauge-fixing of $\mathcal{G}$, then $\hat{\mathcal{L}}(\mathcal{G}')\le\hat{\mathcal{L}}(\mathcal{G})$ and $D(\mathcal{G}')\ge D(\mathcal{G})$.
\end{lem}
\noindent We prove this fact in Appendix~\ref{app:gauge_fixing_distance}.

A concept more general than gauge-fixing is gauge-switching. If both $\mathcal{G}'$ and $\mathcal{G}''$ are gauge-fixings of $\mathcal{G}$, then one can move encoded logical information from $\mathcal{G}'$ to $\mathcal{G}''$ (or vice-versa) while keeping it protected with the stabilizers $\mathcal{S}(\mathcal{G}')\cap\mathcal{S}(\mathcal{G}'')\ge\mathcal{S}(\mathcal{G})$ and with code distance at least $D(\mathcal{G})$. Measuring the gauge group $\mathcal{G}''$, applying a correction based on the values of $\mathcal{S}(\mathcal{G})$ using a decoder for $\mathcal{G}$, and finally projecting onto the $+1$-eigenspaces of elements of $\mathcal{S}(\mathcal{G}'')-\mathcal{S}(\mathcal{G})$ using the appropriate elements of $\mathcal{G}$ achieves this information transfer.

\subsection{Gauge-fixing an aBBS code to a BBS code}
To warm up to Definition~\ref{defn:gauge_fixing}, we show that a BBS code specified by binary matrix $A$ is a gauge-fixing of the aBBS code specified by the same matrix. Of course, these codes do not have the same number of physical qubits, so to make the previous sentence precise we include ancilla qubits to the BBS code. This will be a common occurrence in our gauge-fixing theorems, and so we take a moment to discuss it.

Given a quantum code $\mathcal{Q}$, we will consider appending three types of ancillas: (1) qubits in the $\ket{+}$ state, (2) qubits in the $\ket{0}$ state, and (3) bare gauge qubits denoted $\ket{\perp}$. The new code that includes ancillas is written $\mathcal{Q}\ket{+^{m_+}}\ket{0^{m_0}}\ket{\perp^{m_g}}$ with the number of each type of ancilla indicated. Appending ancillas in this way extends the code's gauge group. Ancillas $\ket{+}$ indicate the inclusion of Paulis $X_i$ into the gauge group for each ancilla index $i$. Likewise, $\ket{0}$ ancillas indicate inclusion of $Z_i$. Bare gauge qubits $\ket{\perp}$ indicate inclusion of both $X_i$ and $Z_i$. 

Now we can formally state the relation between $\text{BBS}(A)$ and $\text{aBBS}(A)$.
\begin{thm}
For all binary matrices ${A\in\mathbb{F}_2^{n_1\times n_2}}$, $\mathcal{Q}'=\text{BBS}(A)\ket{+^{n_1n_2-|A|}}\ket{0^{n_1n_2-|A|}}$ is a gauge-fixing of $\mathcal{Q}=\text{aBBS}(A)$.
\end{thm}
\begin{proof}
We place both codes on the lattices $L_1$ and $L_2$ defined in Section~\ref{sec:aBBS} for the aBBS codes (recall, two $n_1\times n_2$ lattices that share qubits wherever $A_{ij}=1$). The gauge group of $\mathcal{Q}$ is defined in Eqs.~\eqref{eq:ggX_aBBS}, \eqref{eq:ggZ_aBBS}. For $\mathcal{Q}'$, however, we should rewrite the gauge group to fit on these two lattices and to include the ancillas. As one may suspect from their quantity, the $\ket{+}$ ancillas are the type 2 qubits (recall, those in $L_2$ but not in $L_1$) and $\ket{0}$ ancillas are the type 1 qubits (those in $L_1$ but not $L_2$).
\begin{alignat}{2}
\mathcal{G}^{(\mathcal{Q}')}_X&=\big\{X^{(L_1)}(S)X^{(L_2)}(T):&G_RS=0,S\subseteq A\\\nonumber
&&T\subseteq\mathbbm{1}-A\big\},\\
\mathcal{G}^{(\mathcal{Q}')}_Z&=\big\{Z^{(L_1)}(S)Z^{(L_2)}(T):&TG_R^T=0,T\subseteq A\\\nonumber
&&S\subseteq\mathbbm{1}-A\}.
\end{alignat}
Stabilizers of $\mathcal{Q}'$ include not just the stabilizers of $\text{BBS}(A)$, but also single-qubit Pauli $X$s on type 2 qubits and single-qubit Pauli $Z$s on the type 1 qubits. So we have
\begin{align}
\mathcal{S}^{(\mathcal{Q}')}_X&=\big\{X^{(L_2)}(S+T):G_1S=0,SH_R^T=0,T\subseteq\mathbbm{1}-A\big\},\\
\mathcal{S}^{(\mathcal{Q}')}_Z&=\big\{Z^{(L_1)}(S+T):SG_2^T=0,H_RS=0,T\subseteq\mathbbm{1}-A\big\}.
\end{align}

Now it is clear that
\begin{alignat}{3}
\mathcal{G}_X^{(\mathcal{Q}')}&\le\mathcal{G}_X^{(\text{aBBS})},\quad\mathcal{G}_Z^{(\mathcal{Q}')}\hspace{2pt}&\le&\hspace{2pt}\mathcal{G}_Z^{(\text{aBBS})},\\
\mathcal{S}_X^{(\text{aBBS})}&\le\mathcal{S}_X^{(\mathcal{Q}')},\quad\mathcal{S}_Z^{(\text{aBBS})}\hspace{2pt}&\le&\hspace{2pt}\mathcal{S}_Z^{(\mathcal{Q}')}.
\end{alignat}
This takes care of part (1) of Definition~\ref{defn:gauge_fixing}.

Adding ancillas does not change the number of logical qubits in $\mathcal{Q}'$, and so both $\mathcal{Q}'$ and $\mathcal{Q}$ have $\rank{A}$ logical qubits, showing part (2) of Definition~\ref{defn:gauge_fixing} holds.
\end{proof}

\subsection{Gauge-fixing an aBBS code to hypergraph product codes}
In this section, we show that certain hypergraph product codes are gauge-fixings of an aBBS code. Informally, our main result is that for all $A\in\mathbb{F}_2^{n_1\times n_2}$ both $\text{HGP}(H_R,H_2)$ and $\text{HGP}(H_1,H_R)$ are gauge-fixings of $\text{aBBS}(A)$, where we only require that the rows of $H_1\in\mathbb{F}_2^{n_1^T\times n_1}$ and $H_2\in\mathbb{F}_2^{n_2^T\times n_2}$ span $\ker(A^T)$ and $\ker(A)$, respectively.

Just like the case of a BBS code in the last section, to formalize this gauge-fixing we need to define all three of these codes on the same set of physical qubits. Four lattices of qubits are involved, which we label $L_1$, $L_2$, $l_1$, and $l_2$. The code $\text{aBBS}(A)$ is supported on the $n_1\times n_2$ lattices $L_1$ and $L_2$. Recall that qubits in $L_1$ and $L_2$ are identified at the positions where $A_{ij}=1$, so there are $2n_1n_2-|A|$ total qubits in $L_1\cup L_2$.
The code $\text{HGP}(H_R,H_2)$ is supported on lattices $L_1$ and $l_1$, thus making $l_1$ a $(n_1-1)\times n_2^T$ lattice. Similarly, the code $\text{HGP}(H_1,H_R)$ is supported on lattices $L_2$ and $l_2$, and so $l_2$ is a $n_1^T\times(n_2-1)$ lattice. A schematic of this qubit arrangement is shown in Fig.~\ref{fig:lattice_scheme}.

\begin{figure}
\centering
\includegraphics[width=0.45\columnwidth]{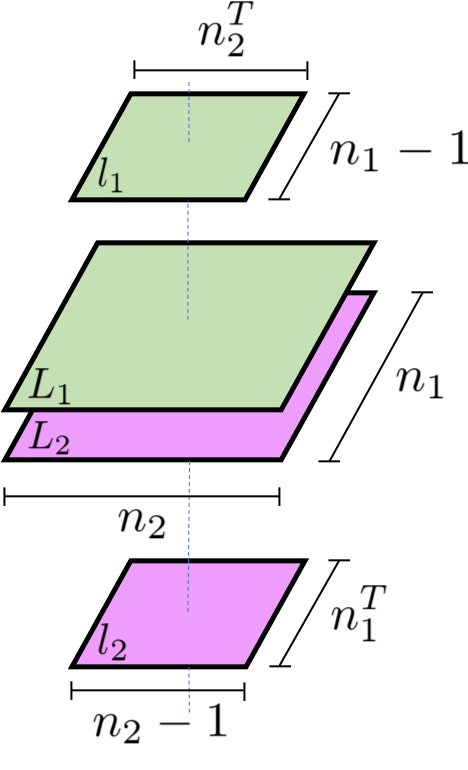}
\caption{\label{fig:lattice_scheme} Gauge-fixing an aBBS code to HGP codes takes place on four lattices of qubits. Each code involved is supported on two lattices -- $\text{aBBS}(A)$ is supported on $L_1$ and $L_2$, $\text{HGP}(H_R,H_2)$ on $L_1$ and $l_1$, and $\text{HGP}(H_1,H_R)$ on $L_2$ and $l_2$.}
\end{figure}

\begin{thm}\label{thm:aBBS_to_HGP}
Let $A\in\mathbb{F}_2^{n_1\times n_2}$ and $H_1\in\mathbb{F}_2^{n_1^T\times n_1}$, $H_2\in\mathbb{F}_2^{n_2^T\times n_2}$ be such that $\row{H_1}=\ker(A^T)$, $\row{H_2}=\ker(A)$. Then the codes
\begin{align}
\mathcal{Q}'&=\text{HGP}(H_{R},H_2)\ket{+^{n_1n_2-|A|}}\ket{\perp^{n_1^T(n_2-1)}},\\
\mathcal{Q}''&=\text{HGP}(H_1,H_{R})\ket{0^{n_1n_2-|A|}}\ket{\perp^{(n_1-1)n_2^T}}
\end{align}
are gauge-fixings of
\begin{equation}
\mathcal{Q}=\text{aBBS}(A)\ket{\perp^{n_1^T(n_2-1)+(n_1-1)n_2^T}}.
\end{equation}
\end{thm}
\begin{proof}
We just prove that $\mathcal{Q}'$ is a gauge-fixing of $\mathcal{Q}$ because proving the same for $\mathcal{Q}''$ is analogous. To prove this half of the theorem, we do not need lattice $l_2$ and so omit it when we write down Pauli operators. Indeed, without $l_2$, code $\mathcal{Q}'$ is a subspace code -- its gauge group is its stabilizer group.

Parity check matrices $H_1$ and $H_2$ define classical codes $\mathcal{C}_1=\col{A}$ and $\mathcal{C}_2=\row{A}$, each encoding $k=\rank{A}$ bits. These codes have some generating matrices $G_1$ and $G_2$ that we will use. Code $\mathcal{C}_2^T$ has a generating matrix $F_2$. By the discussions in Section~\ref{sec:code_background}, we have that both $\mathcal{Q}'$ and $\mathcal{Q}$ encode $k$ qubits, thus verifying part (2) of the gauge-fixing definition, Definition~\ref{defn:gauge_fixing}.

Now, let us write down the stabilizers of $\mathcal{Q}'$ and the gauge group and stabilizers of $\mathcal{Q}$. These follow from the appropriate equations in Section~\ref{sec:code_background}, but with additions due to the ancillas: $\ket{+}$ ancillas in $L_2-L_1$ for $\mathcal{Q}'$ and $\ket{\perp}$ ancillas in $l_1$ for $\mathcal{Q}$.
\begin{widetext}
\begin{align}
\text{By Eq.~\eqref{eq:ShgpX}},\quad\mathcal{S}_X^{(\mathcal{Q}')}&=\big\{X^{(L_1)}(S_1)X^{(L_2)}(S_2)X^{(l_1)}(T):S_1H_2^T=H_R^TT,\text{\space}G_RS_1=0,\text{\space}TF_2^T=0,\text{\space}S_2\subseteq\mathbbm{1}-A\big\},\\
\text{Eq.~\eqref{eq:ShgpZ}},\quad\mathcal{S}_Z^{(\mathcal{Q}')}&=\big\{Z^{(L_1)}(S)Z^{(l_1)}(T):H_RS=TH_2,\text{\space}SG_2^T=0\big\},\\
\text{Eq.~\eqref{eq:ggX_aBBS}},\hspace{3pt}\quad\mathcal{G}_X^{(\mathcal{Q})}&=\big\{X^{(L_1)}(S_1)X^{(L_2)}(S_2)X^{(l_1)}(T):G_RS_1=0,\text{\space}S_2\subseteq\mathbbm{1}-A\big\},\\
\text{Eq.~\eqref{eq:ggZ_aBBS}},\hspace{3pt}\quad\mathcal{G}_Z^{(\mathcal{Q})}&=\big\{Z^{(L_1)}(S_1)Z^{(L_2)}(S_2)Z^{(l_1)}(T):S_2G_R^T=0,\text{\space}S_1\subseteq\mathbbm{1}-A\big\},\\
\text{Eq.~\eqref{eq:SabbsX}},\hspace{3pt}\quad\mathcal{S}_X^{(\mathcal{Q})}&=\big\{X^{(L_2)}(S):SH_R^T=0,\text{\space}G_1S=0\big\},\\
\text{Eq.~\eqref{eq:SabbsZ}},\hspace{3pt}\quad\mathcal{S}_Z^{(\mathcal{Q})}&=\big\{Z^{(L_1)}(S):H_RS=0,\text{\space}SG_2^T=0\big\}.
\end{align}
\end{widetext}
To show part (1) of Definition~\ref{defn:gauge_fixing}, we have four inclusions to prove: (a) $\mathcal{S}_X^{(\mathcal{Q}')}\subseteq\mathcal{G}_X^{(\mathcal{Q})}$, (b) $\mathcal{S}_Z^{(\mathcal{Q})}\subseteq\mathcal{S}_Z^{(\mathcal{Q}')}$, (c) $\mathcal{S}_Z^{(\mathcal{Q}')}\subseteq\mathcal{G}_Z^{(\mathcal{Q})}$, (d) $\mathcal{S}_X^{(\mathcal{Q})}\subseteq\mathcal{S}^{(\mathcal{Q}')}_X$.

Both inclusions (a) and (b) are obvious, so we focus on (c) and (d). For (c), let $M=Z^{(L_1)}(S)Z^{(l_1)}(T)\in\mathcal{S}_Z^{(\mathcal{Q}')}$. Set $S_1=S\cap(\mathbbm{1}-A)$ and $S_2=S\cap A$, so that $M=Z^{(L_1)}(S_1)Z^{(L_2)}(S_2)Z^{(l_1)}(T)$. Now $SG_2^T=0$ implies that rows of $S$ are parity checks of code $\mathcal{C}_2$. Since rows of $A$ are codewords of $\mathcal{C}_2$, each row of $S_2=S\cap A$ contains an even number of 1s. Thus, $S_2=S_2G_R^T=0$, and so $M\in\mathcal{G}_Z^{(\mathcal{Q})}$.

For (d), let $M=X^{(L_2)}(S)\in\mathcal{S}_X^{(\mathcal{Q})}$. Set $S_1=S\cap A$ and $S_2=S\cap(\mathbbm{1}-A)$. Since $G_1S=0$, columns of $S$ are parity checks of $\mathcal{C}_1$. Columns of $A$ are codewords of $\mathcal{C}_1$, and so each column of $S_1$ contains an even number of 1s, or $G_RS_1=0$. Similarly, $SH_R^T=0$ implies that rows of $S$ are codewords of $\mathcal{C}_R$, i.e.~all 1s or all 0s. Therefore, $\row{S_1}\subseteq\row{A}=\mathcal{C}_2$ and $S_1H_2^T=0$. This shows $M=X^{(L_1)}(S_1)X^{(L_2)}(S_2)X^{(l_1)}(0)\in\mathcal{S}_X^{(\mathcal{Q}')}$.
\end{proof}

A special case of Theorem~\ref{thm:aBBS_to_HGP} is the gauge-fixing of the Bacon-Shor code $\text{BBS}(\mathbbm{1})=\text{aBBS}(\mathbbm{1})$ (see Example~\ref{ex:bacon_shor}) to the surface code $\text{HGP}(H_R,H_R)$ (see Example~\ref{ex:surface_bounded} in the Appendix).

Let us conclude this section by briefly discussing the code $\text{HGP}(H_1,H_R)$ that we just showed is a gauge-fixing of $\text{aBBS}(A)$. In particular, we would like to argue that it has an asymptotic threshold when $H_1$ is chosen appropriately. Kovalev and Pryadko \cite{Kovalev2013} have shown that any $\llbracket N,K,D\rrbracket$ quantum code family that is $(\beta,\gamma)$-LDPC for constants $\beta$ and $\gamma$ and has distance scaling at least logarithmically in code size, i.e.~$D=\Omega(\log N)$, possesses an asymptotic threshold. Say that $H_1$ is a full-rank, $(b,c)$-LDPC set of parity checks for code $\mathcal{C}_1$ with parameters $[n,k,d]$ and $H_R$ represents the length $n$ repetition code. Then, $\text{HGP}(H_1,H_R)$ is $(\gamma,\gamma)$-LDPC for $\gamma=\max(b,c)+2$ and has parameters $\llbracket N,k,d\rrbracket$ with $N\le 2n^2$. Clearly then, if $\mathcal{C}_1$ is an LDPC code family with $d$ scaling at least logarithmically in $n$, i.e.~$d=\Omega(\log n)$, then by \cite{Kovalev2013} the quantum code family $\text{HGP}(H_1,H_R)$ has an asymptotic threshold.

\section{Discussion}\label{sec:discussion}
We have presented another connection between classical and quantum error-correction and discussed one of its consequences, the construction of Bravyi-Bacon-Shor subsystem codes that are local in 2-dimensions and have optimal parameters. We also showed a somewhat surprising connection between Bravyi-Bacon-Shor codes and the hypergraph product codes via the process of gauge-fixing. 

We briefly point out two somewhat obvious but interesting properties of any gauge-fixing $\mathcal{Q}'$ of Bravyi-Bacon-Shor codes, including e.g.~$\text{HGP}(H_1,H_R)$. First, if the Bravyi-Bacon-Shor codes are optimal, then $\mathcal{Q}'$ is not local in 2-dimensions. This is necessarily the case because if an $\llbracket N,K,D\rrbracket$ subsystem code local in 2-dimensions can be gauge-fixed to a $\llbracket N,K,D'\rrbracket$ subspace code ($D'\ge D$ by Lemma~\ref{lem:gauge_fixing_distance}) local in 2-dimensions, then $KD'^{2}=O(N)$ by \cite{Bravyi2010} implying that $KD<KD^2\le O(N)$, i.e.~the subsystem code is suboptimal. This is also why the 2-dimensional ``topological" subsystem codes (see e.g.~\cite{Bombin2010,Suchara2011,Andrist2012,Sarvepalli2012,Bravyi2013}), which are defined by having stabilizer groups that are local in 2-dimensions, cannot actually compete, despite being subsystem codes, for the $KD=O(N)$ bound.

Second, $\mathcal{Q}'$ does not have constant rate. Indeed, simply rearranging the subsystem bound we get $K/N=O(1/D)$, which vanishes provided the code family has growing distance. Thus, it is impossible to gauge-fix Bravyi-Bacon-Shor codes to hypergraph product codes with constant rate, which is interesting because obtaining constant rate quantum codes is one of the most notable properties of the general-case hypergraph product construction \cite{Tillich2014}. Instead, we necessarily ended up gauge-fixing to a special case $\text{HGP}(H_1,H_R)$ without constant rate.

On the other hand, one of the interesting consequences of our results is the ability to gauge-switch between several hypergraph product codes. For example, one can switch between $\text{HGP}(H_R,H_2)$ and $\text{HGP}(H_1,H_R)$ for any $H_1$ and $H_2$ or between $\text{HGP}(H_1,H_R)$ and $\text{HGP}(H_1',H_R)$ where $H_1$ and $H_1'$ are different parity check matrices for the same classical code. In the process, encoded data is protected by the underlying augmented Bravyi-Bacon-Shor code (see Theorem~\ref{thm:aBBS_to_HGP}), which has the same code distance as the hypergraph product codes in question although it lacks an asymptotic threshold. Nonetheless, generalizing this gauge-switching idea to more hypergraph product codes would be an interesting extension of our work here.

\section*{Acknowledgements}
The author gratefully acknowledges helpful discussions with Sergey Bravyi, Ken Brown, Chris Chamberland, and Andrew Cross. Partial support for this project was generously provided by the IBM Research Frontiers Institute.

\appendix

\section{Hypergraph product codes}\label{app:hypergraph_product_codes}
In this appendix, we review the original presentation of hypergraph product codes \cite{Tillich2014} and verify that our description in Section~\ref{sec:HGP} is equivalent. We also review the derivation of the hypergraph product code parameters. Mainly, our arguments are similar to those in \cite{Tillich2014} and \cite{Campbell2018}.

Recall that the input to the construction is two parity check matrices $H_1\in\mathbb{F}_2^{n_1^T\times n_1}$ and $H_2\in\mathbb{F}_2^{n_2^T\times n_2}$. These have corresponding full-rank generating matrices $G_1\in\mathbb{F}_2^{k_1\times n_1}$ and $G_2\in\mathbb{F}_2^{k_2\times n_2}$ for the classical codes $\mathcal{C}_1$ and $\mathcal{C}_2$. Without loss of generality, we assume $k_1,k_2>0$. Additionally, there are full-rank generating matrices $F_1\in\mathbb{F}_2^{k_1^T\times n_1^T}$ and $F_2\in\mathbb{F}_2^{k_2^T\times n_2^T}$ for the transpose classical codes $\mathcal{C}^T_1$ and $\mathcal{C}^T_2$.

In the original description, the supports of Pauli operators are specified by vectors from $\mathbb{F}_2^{N}$ with $N=n_1n_2+n_1^Tn_2^T$. Generating sets of $X$- and $Z$-type stabilizers are presented as rows of matrices:
\begin{align}\label{eq:ShgpX_gen}
S_X&=\left(\begin{array}{cc}H_1\otimes I_{n_2}&I_{n_1^T}\otimes H_2^T\end{array}\right),\\\label{eq:ShgpZ_gen}
S_Z&=\left(\begin{array}{cc}I_{n_1}\otimes H_2&H_1^T\otimes I_{n_2^T}\end{array}\right),
\end{align}
where $I_n$ is the $n\times n$ identity matrix. That is, if $X^{\vec v}=\prod_{i=1}^{N}X_i^{\vec v_i}$ and we wanted to write out the entire sets of Pauli stabilizers, we would have
\begin{align}
\mathcal{S}_X&=\big\{X^{\vec v}:\vec v\in\row{S_X}\big\},\\
\mathcal{S}_Z&=\big\{Z^{\vec u}:\vec u\in\row{S_Z}\big\}.
\end{align}
It is easy to see that these stabilizers commute, because $S_XS_Z^T=0$. Moreover, from the generating sets in Eqs.~\eqref{eq:ShgpX_gen}, \eqref{eq:ShgpZ_gen}, we note that using classical LDPC parity checks $H_1$ and $H_2$ lead to a quantum LDPC code with the appropriate parameters from Eqs.~\eqref{eq:hgpBeta}, \eqref{eq:hgpGamma}.

We can calculate the number of encoded qubits by finding the number of independent stabilizer generators $\rank{S_X}+\rank{S_Z}$ and subtracting that from $N$. Basic linear algebra says
\begin{equation}
\rank{S_X}=\rank{S_X^T}=n_1^Tn_2-\dim(\ker(S_X^T)).
\end{equation}
Since
\begin{equation}
S_X^T=\left(\begin{array}{c}H_1\otimes I_{n_2}\\I_{n_1^T}\otimes H_2^T\end{array}\right)
\end{equation}
has kernel
\begin{equation}
\ker(S_X^T)=\{x\otimes y:x\in\mathcal{C}_1,y\in\mathcal{C}_2^T\},
\end{equation}
we see that $\dim(\ker(S_X^T))=\dim(\mathcal{C}_1)\dim(\mathcal{C}_2^T)=k_1k_2^T$. A similar argument holds for $S_Z$. Thus, we have
\begin{align}
\rank{S_X}&=n_1^Tn_2-k_1k_2^T,\\
\rank{S_Z}&=n_1n_2^T-k_1^Tk_2.
\end{align}
Accordingly, the hypergraph product code encodes
\begin{align}
K&=N-(n_1^Tn_2-k_1k_2^T)-(n_1n_2^T-k_1^Tk_2)\\
&=(n_1-n_1^T)(n_2-n_2^T)+k_1k_2^T+k_1^Tk_2\\
&=(k_1-k_1^T)(k_2-k_2^T)+k_1k_2^T+k_1^Tk_2\\
&=k_1k_2+k_1^Tk_2^T
\end{align}
qubits. For the third equality, we used Eq.~\eqref{eq:transpose_relation}. This verifies Eq.~\eqref{eq:hgpK}.

Let us create a generating set of logical operators for these qubits. We notice that
\begin{align}\label{eq:LX}
L_X&=\left(\begin{array}{cc}
H_1\otimes I_{n_2}&I_{n_1^T}\otimes H_2^T\\
I_{n_1}\otimes G_2&0\\
0&F_1\otimes I_{n_2^T}
\end{array}\right),\\\label{eq:LZ}
L_Z&=\left(\begin{array}{cc}
I_{n_1}\otimes H_2&H_1^T\otimes I_{n_2^T}\\
G_1\otimes I_{n_2}&0\\
0&I_{n_1^T}\otimes F_2
\end{array}\right)
\end{align}
do in fact provide sets of logical operators because $S_ZL_X^T=0$ and $S_XL_Z^T=0$ demonstrate the appropriate commutation. 

To show that these are indeed complete sets of logical operators, we can calculate the rank of $C=L_XL_Z^T$, which encodes how the $X$- and $Z$-type logical operators commute. There should be $K$ independent, anti-commuting pairs of logical operators, so the rank of $C$ should be $K$. Since
\begin{equation}
C=\left(\begin{array}{ccc}
0&0&0\\
0&G_1^T\otimes G_2&0\\
0&0&F_1\otimes F_2^T
\end{array}\right),
\end{equation}
we do have
\begin{align}
\rank{C}&=\rank{G_1}\rank{G_2}+\rank{F_1}\rank{F_2}\\
&=k_1k_2+k_1^Tk_2^T=K.
\end{align}
We also point out that the last rows of $L_X$ and $L_Z$ (those involving $F_1$ and $F_2$) only contain nontrivial logical operators if both $F_1$ and $F_2$ are nontrivial matrices (i.e.~both $k_1^T$ and $k_2^T$ are greater than zero).

Now consider ``reshaping" \cite{Campbell2018} the vectors that represent Paulis into matrices. Let $\hat e_i$ be the unit vector $(\hat e_i)_j=\delta_{ij}$. A vector $\vec s\in\mathbb{F}_2^{n_1n_2}$ can be decomposed as
\begin{equation}\label{eq:vec_reshape}
\vec s=\sum_{i=1}^{n_1}\sum_{j=1}^{n_2}S_{ij}\hat e_i\otimes\hat e_j
\end{equation}
where $S$ is the matrix corresponding to $\vec s$ and the support of Pauli $X^{\vec s}$ once we have placed it on the $n_1\times n_2$ lattice $L$. We previously wrote this Pauli as $X^{(L)}(S)$. Likewise, vectors $\vec t\in\mathbb{F}_2^{n_1^Tn_2^T}$ are reshaped to represent Paulis on the $n_1^T\times n_2^T$ lattice $l$.

Linear transformations of $\vec s$ correspond to matrix multiplications on $S$. By Eq.~\eqref{eq:vec_reshape},
\begin{equation}\label{eq:lin_reshape}
(U\otimes V)\vec s\longmapsto USV^T.
\end{equation}
Likewise with transformations on $\vec t$.

At this point we can justify our presentation of the stabilizers and logical operators, Eqs.~(\ref{eq:ShgpX}, \ref{eq:ShgpZ}) and (\ref{eq:LhgpX}, \ref{eq:LhgpZ}) in the main text. We can characterize elements of $\row{S_X}$ by the fact that they commute with all rows of $L_Z$.
\begin{equation}
\left(\begin{array}{c}\vec s\\\vec t\end{array}\right)\in\row{S_X}\text{\space\space iff\space\space}L_Z\left(\begin{array}{c}\vec s\\\vec t\end{array}\right)=\vec 0.
\end{equation}
Reshaping the linear equations on the right using Eq.~\eqref{eq:lin_reshape} gives the equations
\begin{align}
SH_2^T=H_1^TT,\quad G_1S=0,\quad TF_2^T=0,
\end{align}
which are exactly the conditions on $S$ and $T$ in $\mathcal{S}^{(\text{hgp})}_X$, Eq.~\eqref{eq:ShgpX}.

Similarly, elements of $\row{S_Z}$ are characterized by commutation with rows of $L_X$, elements of $\row{L_X}$ by commutation with rows of $S_Z$, and elements of $\row{L_Z}$ by commutation with rows of $S_X$. After reshaping the appropriate linear equations, one can confirm Eqs.~(\ref{eq:ShgpZ}, \ref{eq:LhgpX}, \ref{eq:LhgpZ}).

Finally, we prove that the hypergraph product code has the claimed distance from Eq.~\eqref{eq:hgpD}. We begin by bounding the weight of nontrivial $X$-type logical operators, those elements of $\mathcal{L}_X^{(\text{hgp})}-\mathcal{S}_X^{(\text{hgp})}$. If $M=X^{(L)}(S)X^{(l)}(T)$, then $SH_2^T=H_1^TT$ and there is an $M'\in\mathcal{L}_Z^{(\text{hgp})}-\mathcal{S}_Z^{(\text{hgp})}$ that anticommutes with $M$. In fact, given the basis in $L_Z$, Eq.~\eqref{eq:LZ}, we know something about the form of $M'$ -- it corresponds either to a row of $G_1\otimes I_{n_2}$ (case (1)) or, if $k_1^T,k_2^T>0$, to a row of $I_{n_1^T}\otimes F_2$ (case (2)).

In case (1), we can take $M'=X^{(L)}(S')$ where $S'$ is an outer product $S'=\vec c\hspace{2pt}\hat e_j^T$ for some $\vec c\in\mathcal{C}_1$ and some $j$. As $M$ and $M'$ anticommute,
\begin{equation}
1=\tr\left(S^TS'\right)=\hat e_j^TS^T\vec c
\end{equation}
and clearly $S^T\vec c\neq\vec 0$. Now, $H_2S^T\vec c=T^TH_1\vec c=0$ and thus $S^T\vec c$ is a nonzero vector in $\ker(H_2)=\mathcal{C}_2$. Therefore, $|M|\ge|S|=|S^T|\ge|S^T\vec c|\ge d_2$.

In case (2), which is relevant only if $k_1^T, k_2^T>0$, the argument is analogous. Take $M'=X^{(l)}(T')$ where $T'$ is the outer product $T'=\hat e_i\hspace{1pt}\vec b^T$ for some $\vec b\in\mathcal{C}_2^T$ and some $i$. As $M$ and $M'$ anticommute,
\begin{equation}
1=\tr\left(T^TT'\right)=\vec b^TT^T\hat e_i
\end{equation}
and clearly $\vec b^TT^T\neq\vec 0$. Also, $\vec b^TT^TH_1=\vec b^TH_2S^T=\vec 0^T$ and so $T\vec b$ is a nonzero vector in $\ker(H_1^T)=\mathcal{C}_1^T$. Thus, $|M|\ge|T|\ge|T\vec b|\ge d_1^T$.

From these two cases, we conclude 
\begin{equation}
|M|\ge\bigg\{\begin{array}{lr}d_2,&k_1^T=0\text{ or }k_2^T=0\\\min(d_1^T,d_2),&\text{otherwise}\end{array}.
\end{equation}
If we go through the analogous argument for nontrivial $Z$-type logical operators, we would find their weight bounded below by $d_1$ in the case that one of $\mathcal{C}_1^T$ or $\mathcal{C}_2^T$ is trivial and $\min(d_1,d_2^T)$ otherwise. Thus, the code distance of the hypergraph product code is
\begin{equation}
D\ge\bigg\{\begin{array}{lr}\min(d_1,d_2),&k_1^T=0\text{ or }k_2^T=0\\\min(d_1,d_2,d_1^T,d_2^T),&\text{otherwise}\end{array}.
\end{equation}
By looking at $L_X$ and $L_Z$, Eqs.~\eqref{eq:LX} and \eqref{eq:LZ}, we see that there are indeed logical operators saturating this inequality, and so we have verified Eq.~\eqref{eq:hgpD}.

We conclude this appendix by reviewing the surface code as a special case of the hypergraph product. In fact, there are two versions of the surface code that can be made: the one with boundary \cite{Bravyi1998} and the one on a torus \cite{Kitaev2003}.

\begin{example}\label{ex:surface_bounded}
The surface code with boundary \cite{Bravyi1998} is an $\llbracket n^2+(n-1)^2,1,n\rrbracket$ code. These parameters match those of $\text{HGP}(H_R,H_R)$. Indeed, we draw some of the stabilizers indicated by rows of $S_X$ and $S_Z$ in Fig.~\ref{fig:surface_code}(a), in which one can recognize the surface code.
\end{example}

\begin{example}\label{ex:surface_torus}
The surface code on a torus \cite{Kitaev2003} is a $\llbracket 2n^2,2,n\rrbracket$ code, matching the parameters of $\text{HGP}(H_R',H_R')$ for
\begin{equation}
H_R'=\left(\begin{array}{ccccccc}
1&1&0&\dots&0&0&0\\
0&1&1&0&\dots&0&0\\
&\ddots&&\ddots&&\ddots&\\
0&0&0&\dots&0&1&1\\
1&0&0&\dots&0&0&1
\end{array}\right).
\end{equation}
This is an over-complete parity check matrix for the $[n,1,n]$ classical repetition code -- the sum of all rows is $\vec0$. Notice the transpose code is also the $[n,1,n]$ repetition code. We draw some of the stabilizers corresponding to rows of $S_X$ and $S_Z$ in Fig.~\ref{fig:surface_code}(b) in which one can recognize the surface code on the torus.
\end{example}

\begin{figure}
\includegraphics[width=\columnwidth]{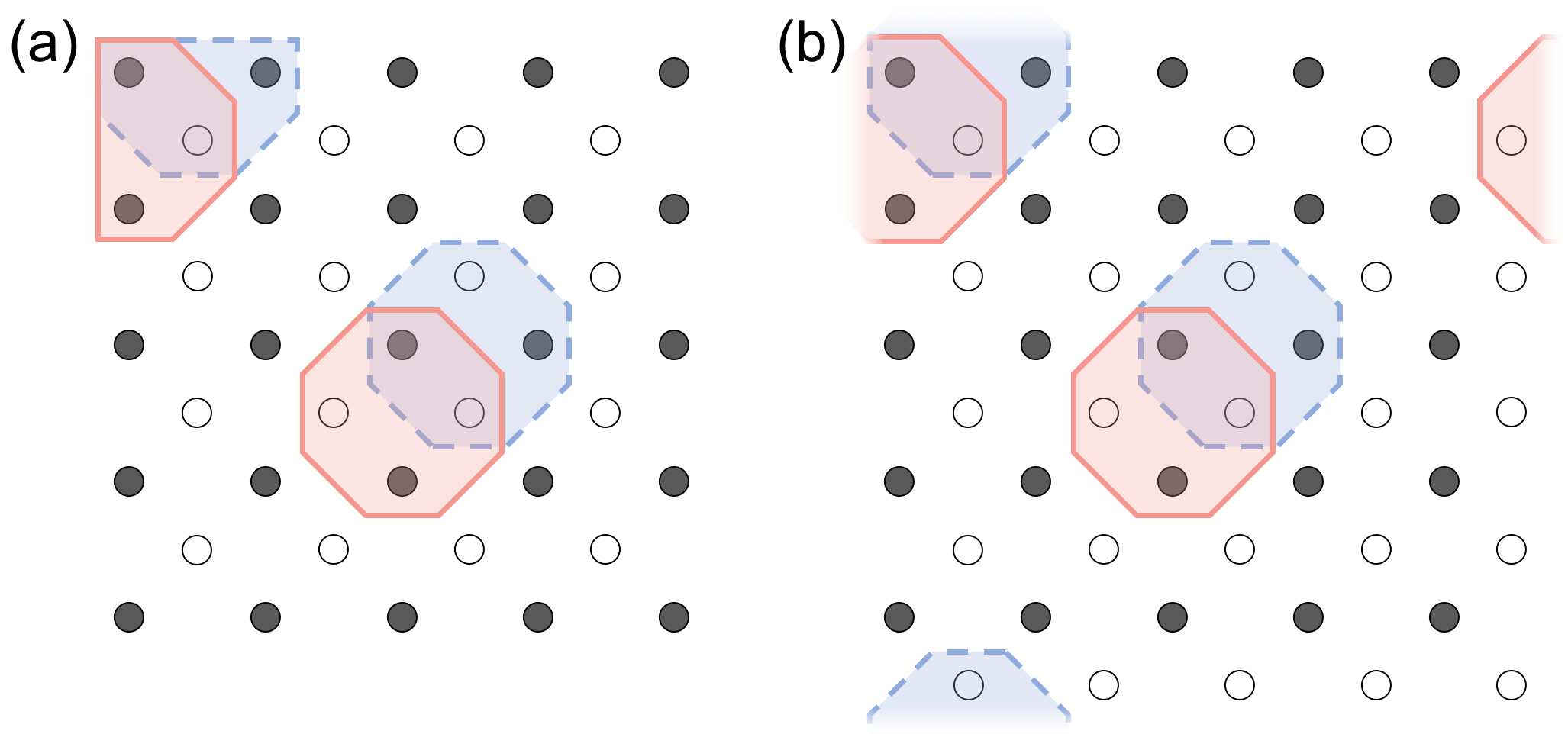}
\caption{\label{fig:surface_code} The surface code (a) with boundary and (b) on the torus drawn on the $L$ (filled qubits) and $l$ (unfilled qubits) lattices of the hypergraph product. Some example $X$- (red, solid) and $Z$-type (blue, dashed) stabilizers are shown. These example stabilizers correspond to select rows of the matrices $S_X$ and $S_Z$ of the appropriate hypergraph products.}
\end{figure}


\section{Expander codes}\label{app:expander_codes}
Constructions of good families of classical LDPC codes based on expander graphs are known. In this section, we review the segment of expander theory that is needed to prove the goodness of these codes, and therefore the goodness of Bravyi-Bacon-Shor codes constructed from them. All of this section is classical and we expect to do no more than inform any uninitiated readers of what is already known.

\subsection{Construction}\label{app:expander_codes_construction}
The objects used to construct good classical LDPC codes are called lossless-expanders \cite{Capalbo2002}, though we will refer to them simply as expanders. Mathematically, these expanders are undirected, bipartite graphs, which we will represent by a tuple $(L,R,E)$ of left nodes, right nodes, and edges. A node $v$ has a degree, the number of edges incident to it, which we denote $\deg{v}$. Given a set of nodes $V\subseteq L\cup R$, we can talk about its set of neighbors
\begin{equation}
\Gamma(V)=\{u:\exists v\in V\text{ s.t.~}(u,v)\in E\}.
\end{equation}
Expanders attempt to maximize the size of $\Gamma(S)$ for all $S\subseteq L$ sufficiently small.

\begin{defn}[Expanders]
A $(n,m,b,\delta,\epsilon)$ expander is a bipartite graph $(L,R,E)$ satisfying
\begin{enumerate}
\item[]\hspace{-10pt}Size: $|L|=n$, $|R|=m$,
\item[]\hspace{-10pt}Degree: $\forall v\in L$, $\deg(v)=b$, $\forall w\in R$, $\deg(w)=c=nb/m$,
\item[]\hspace{-10pt}Expansion: $\forall S\subseteq L$ s.t.~$|S|\le(1-\delta)n$, $(1-\epsilon)b|S|\le|\Gamma(S)|\le b|S|$.
\end{enumerate}
\end{defn}
\noindent In particular, expanders with smaller $\delta$ and $\epsilon$ are better than those with larger values. The expansion property is trivial if $\delta>1-2/n$ for instance. Moreover, if $\delta\le1-2/n$ and $\epsilon=0$, only the graph with $m=nb$ right-nodes and $n$ connected components suffices to meet the definition. Finally, $b=1$ or $c=nb/m=1$ lead to similar trivialities. Thus, we take $\delta\le1-2/n$, $\epsilon>0$, and $b,c>1$ throughout.

From the definition, one can prove other facts about expanders. One very useful fact for us concerns the size of the set of ``unique" neighbors of $V\subseteq L$,
\begin{equation}
\Gamma_1(V)=\{u\in\Gamma(V):|\Gamma(\{u\})\cap V|=1\}.
\end{equation}
Elements of $\Gamma_1(V)$ are the elements of $\Gamma(V)$ that have just one neighbor in $V$.
\begin{lem}\label{lem:unique_neighbors}
Suppose the bipartite graph $(L,R,E)$ is an $(n,m,b,\delta,\epsilon)$ expander and $S\subseteq L$ satisfies $|S|\le(1-\delta)n$. Then
\begin{equation}
\Gamma_1(S)\ge(1-2\epsilon)b|S|.
\end{equation}
\end{lem}
\begin{proof}
The number of edges leaving $S$ is $b|S|$. This is the same as the number of edges entering $S$ from $\Gamma(S)$. The nodes in $\Gamma_1(S)\subseteq\Gamma(S)$ have exactly 1 such edge, while those in $\Gamma_{\ge2}(S)=\Gamma(S)-\Gamma_1(S)$ have at least 2 such edges. Thus,
\begin{align}
b|S|&\ge 2|\Gamma_{\ge2}(S)|+|\Gamma_1(S)|\\
&=|\Gamma_{\ge2}(S)|+|\Gamma(S)|\\
&=2|\Gamma(S)|-|\Gamma_1(S)|\\
&\ge2(1-\epsilon)b|S|-|\Gamma_1(S)|,
\end{align}
where the last inequality uses the expansion property.
\end{proof}

To create a classical code from an expander, we will use (the simplest version of) Tanner's construction \cite{Tanner1981}. This prescribes that we view the left nodes $L$ as a set of code bits and each right node as specifying a parity check on the bits that are its neighbors. More precisely, define the incidence matrix $\Lambda\in\mathbb{F}_2^{|L|\times |R|}$ of a bipartite graph $G=(L,R,E)$ as
\begin{equation}
\Lambda_{uv}=\bigg\{\begin{array}{ll}
0,&(u,v)\not\in E\\
1,&(u,v)\in E
\end{array}.
\end{equation}
Then, $H=\Lambda^T$ takes the role of a parity check matrix to define the Tanner code of $G$, $\mathcal{C}_G=\ker(H)$.

If $G$ is an expander, we call $\mathcal{C}_G$ an expander code. In this case, we can place useful bounds on its code parameters.
\begin{lem}\label{lem:expander_code_parameters}
Suppose $G=(L,R,E)$ is an $(n,m,b,\delta,\epsilon)$ expander with $\epsilon<1/2$. Then $\mathcal{C}_G$ is a $[n,k,d]$ code with $k\ge n-m$ and $d\ge2(1-\epsilon)\lfloor(1-\delta)n\rfloor$.
\end{lem}
\begin{proof}
The parity check matrix $H$ of $\mathcal{C}_G$ has $m$ rows, and thus its kernel is at least $n-m$ dimensional. So, $k\ge n-m$.

Let $\vec s\in\mathbb{F}_2^n$ be a bit string and $S=\{v:\vec s_v=1\}\subseteq L$ be its support. We show that if $|\vec s|=|S|<2(1-\epsilon)\lfloor(1-\delta)n\rfloor$, then there must be a parity check unsatisfied by $\vec s$, and so $\vec s$ is not a codeword. To do this, it is sufficient to show that $\Gamma_1(S)$ is not empty -- any $w\in\Gamma_1(S)$ cannot be a satisfied check as only a single bit in the check is 1.

Suppose first that $|S|\le(1-\delta)n$. Then by Lemma~\ref{lem:unique_neighbors}, we have $|\Gamma_1(S)|\ge(1-2\epsilon)b|S|>0$, using the assumption $\epsilon<1/2$.

Now suppose $(1-\delta)n<|S|<2(1-\epsilon)\Delta$ where $\Delta=\lfloor(1-\delta)n\rfloor$. Let $T\subseteq S$ satisfy $|T|=\Delta$. So,
\begin{equation}\label{eq:T_bound}
\Gamma_1(T)\ge(1-2\epsilon)b\Delta
\end{equation}
by Lemma~\ref{lem:unique_neighbors}. At the same time $|S-T|=|S|-|T|<(1-2\epsilon)\Delta<\Delta$ implies
\begin{equation}\label{eq:S-T_bound}
|\Gamma(S-T)|<(1-2\epsilon)b\Delta,
\end{equation} 
because nodes in $S-T$ are degree $b$. A check $w$ is in $\Gamma_1(S)$ if $w\in\Gamma_1(T)$ and $w\not\in\Gamma(S-T)$. Since $|\Gamma_1(T)|>|\Gamma(S-T)|$ by Eqs.~\eqref{eq:T_bound}, \eqref{eq:S-T_bound}, we have $|\Gamma_1(S)|>0$.
\end{proof}

It is worth noting when a family of expander codes $[n,k,d]$ is good, i.e.~$k=\Theta(n)$ and $d=\Theta(n)$. Using Lemma~\ref{lem:expander_code_parameters}, it is sufficient that $\epsilon,\delta$ are constant (independent of $n$) and that $m/n=b/c$ is constant. It is typical to construct families in which $b$ (the degree of nodes on the left) and $c$ (the degree of nodes on the right) are both constant individually. This makes the code a low-density parity check code and also enables the efficient decoder discussed in the next section.

Lemma~\ref{lem:expander_code_parameters} assumes $\epsilon<1/2$ which means it is only sufficient for analyzing expander codes constructed from expanders with sufficiently large expansion. For a long time, although expanders of arbitrarily large size with $\epsilon<1/2$ were known to exist by counting, it was not known how to construct them. However, the zig-zag construction \cite{Capalbo2002} eventually solved this problem. For our purposes, a suitable distillation of their result is the following.
\begin{thm}[Hoory, Linial, Wigderson \cite{Hoory2006}, Thm.~10.4]\label{thm:zigzag_construction}
For every $\epsilon>0$ and $\alpha\in(0,1)$, there exist constants $\gamma,\sigma$ and an explicit family of $(n,m,b,\delta,\epsilon)$ expanders with $m=\alpha n$,
\begin{align}
b&\le\left(\frac{1}{\epsilon \alpha}\right)^\gamma,\\
\delta&\le1-\sigma\left(\epsilon \alpha\right)^{\gamma+1}.
\end{align}
Using Lemma~\ref{lem:expander_code_parameters}, the corresponding expander codes have parameters $[n,k,d]$ with
\begin{align}
k&\ge(1-\alpha)n,\\
d&\ge2(1-\epsilon)\lfloor\sigma(\epsilon \alpha)^{\gamma+1}n\rfloor.
\end{align}
\end{thm}

Theorem~\ref{thm:zigzag_construction} is a theoretically important result -- it provides a construction of a good family of classical codes, and moreover the parity checks involve only constant numbers of bits. However, the constants involved may not be the most practical, and random instances of bipartite graphs, like those analyzed in the Appendix of \cite{Sipser1996} or in Theorem~8.7 of \cite{Richardson2008}, may be less cumbersome to work with.

\subsection{Decoding}\label{app:expander_codes_decoding}
Sipser and Spielman \cite{Sipser1996} analyzed a decoder for classical expander codes that operates in greedy fashion by flipping any bits that overall reduce the number of unsatisfied parity checks. We will refer to this as the flip decoder. They show that for expanders with sufficiently large expansion ($\epsilon<1/4$) the flip decoder corrects any number of errors within a constant fraction of the code distance and does so in time proportional to the code size, i.e.~in linear time. Later Spielman \cite{Spielman1996} analyzed the flip decoder in the scenario that the parity checks are noisy in addition to the bits. It is this latter scenario that is most relevant to the quantum case where we may only noisily measure parity checks and not the data qubits themselves. We provide a somewhat generalized presentation of Spielman's analysis here. In particular, we show that for expanders with larger expansion (smaller $\epsilon$) the flip decoder deals with measurement errors better.

Let $\hat e_i$ denote the vector with elements $(\hat e_i)_j=\delta_{ij}$. Here it represents a flip of the $i^{\text{th}}$ bit. The flip decoder is defined as follows.
\begin{defn}[Sipser-Spielman Flip Decoder \cite{Sipser1996,Spielman1996}]\label{defn:spielman_noisy_flip_decoder}
Given an expander code $\mathcal{C}$ with parity check matrix $H\in\mathbb{F}_2^{m\times n}$ and a vector indicating unsatisfied checks $\vec u\in\mathbb{F}_2^m$, return a set of corrections $\pvec e'\in\mathbb{F}_2^n$ by doing the following.
\begin{enumerate}[label=(\arabic*)]
\item Initialize $\pvec e'=0^n$ and $\pvec u'=\vec u$.
\item Repeat
\begin{enumerate}
\item Find $i\in\{1,2\dots,n\}$ such that $|\pvec u'|>|H\hat e_i-\pvec u'|$. If none exists, return $\pvec e'$.
\item Let $\pvec e'\leftarrow \pvec e'+\hat e_i$ and $\pvec u'\leftarrow H\hat e_i-\pvec u'$.
\end{enumerate}
\end{enumerate}
Steps (2a) and (2b) constitute a decoding ``round".
\end{defn}

Since the number of unsatisfied checks $|\vec u'|$ decreases each round and there are $O(n)$ checks in a $[n,k,d]$ expander code, it is somewhat reasonable to believe that this decoder takes linear time. 
\begin{lem}[Sipser and Spielman \cite{Sipser1996}]\label{lem:flip_linear_time}
Let $\mathcal{C}$ be an $[n,k,d]$ expander code based on an $(n,m,b,\delta,\epsilon)$ expander graph with $b$ and $m/n$ constant. The flip decoder for $\mathcal{C}$ runs in time $O(n)$.
\end{lem}
\begin{proof}
Proving this simply requires a suitable data structure. We assume that the adjacency matrix of the expander (or equivalently the check matrix of the code) is given in a sparse matrix representation, so it takes constant time to obtain a list of neighbors of a bit or check in the expander graph.

Recall $\vec u\in\mathbb{F}_2^m$ is given as the value of the $m$ parity checks. At the beginning of the decoding, we calculate for each bit $i$ the number $v_i$ of unsatisfied checks that it is involved in. This takes $O(bn)=O(n)$ total time. We construct $b+1$ linked lists, one for each possible value of $v_i$, and place each $i$ in the corresponding list. That is, for each $i\in\{1,2,\dots,n\}$, we store $\{i,v_i,p_i,n_i\}$, where $p_i,n_i\in\{1,2,\dots,n\}$ point to the previous and next elements in the linked list (or are null if $i$ is at the head or tail). Variables $h_v\in\{1,2,\dots,n\}$ for every $v\in\{0,1,\dots,b\}$ point to the linked list heads (or null if the list is empty). The initial setup of $p_i$, $n_i$, and $h_v$ values takes $O(n)$ time. It is also important to note that removing from and attaching to the front of linked lists take $O(1)$ time. 

The main body of the flip decoding algorithm is the iteration in Step (2) of Definition~\ref{defn:spielman_noisy_flip_decoder}. Since the number of unsatisfied clauses strictly decreases during each round, there are at most $O(m)=O(n)$ rounds. Moreover, each round can be made to take constant time, as we now show.

Every round the algorithm begins by finding the non-null $h_v$ with largest $v$. This takes $O(b)$ time. If $0\le v\le b/2$, then there is no bit to flip to reduce the number of unsatisfied clauses and the algorithm returns. If $v>b/2$, then flip bit $h_v$. This causes $b=O(1)$ checks $j$ to flip and we update the values $u_j$ accordingly. Within each of the flipped checks are $c=O(1)$ bits $i$ which now participate in either one more or one fewer unsatisfied check. The values $v_i$ should be updated accordingly and the linked list element $\{i,v_i,p_i,n_i\}$ removed from its current linked list and inserted at the head of list $h_{v_i}$, which takes $O(1)$ time. Thus, the entire round takes $O(1)$ time.
\end{proof}

Presently, we concern ourselves with how well the decoder corrects errors. The main result is that the number of errors on the data can be reduced to a constant fraction of the number of errors on the checks.

\begin{thm}[Spielman \cite{Spielman1996}]\label{thm:flip_decoder}
Let $\mathcal{C}$ be an expander code constructed from a $(n,m,b,\delta,\epsilon)$ expander with $\epsilon<\frac14-\frac{r}{b}$ for $1\le r<b/4$. Given input $\vec u=H(\vec s_0+\vec e)+\vec f$ for $\vec s_0\in\mathcal{C}$ and provided
\begin{equation}\label{eq:noisy_flip_assumption}
|\vec e|+\frac{2}{b}|\vec f|\le(1-2\epsilon)\lfloor(1-\delta)n\rfloor,
\end{equation}
the noisy flip decoder returns $\pvec e'$ such that $|\pvec e'-\vec e|<|\vec f|/r$.
\end{thm}
\begin{proof}
Let $E=\{i:\vec e_i+\pvec e'_i=1\}$ be the set of corrupted message bits and $U=\{j:\pvec u'_j=1\}$ be the set of unsatisfied checks at any point during execution of the algorithm. Let $S=\Gamma(E)-U$ be the satisfied checks in the neighborhood of $E$. Provided $|E|=|\pvec e'-\vec e|\le(1-\delta)n$, the expansion property implies
\begin{equation}\label{eq:noisy_flip_lemma_ineq1}
|U|+|S|\ge|\Gamma(E)|\ge(1-\epsilon)b|E|.
\end{equation}
This gives a lower bound on $|U|$ and $|S|$. 

We can get an upper bound on these by a counting argument. Imagine we add $m$ additional nodes to the left side of the bipartite expander graph and connect these new nodes pairwise to the corresponding $m$ check nodes on the right side. These new nodes represent the presence (if set to $1$) or absence (if set to $0$) of an error on the check bit. So, of these new nodes, $|\vec f|$ are set to 1, those in the set $F=\{j+n:\vec f_j=1\}$. Now every check in $U$ is connected to at least one node in $E\cup F$ and every check in $S$ is connected to at least two nodes in $E\cup F$. Since there are $b|E|+|\vec f|$ edges leaving $E\cup F$, we have
\begin{equation}\label{eq:noisy_flip_lemma_ineq2}
b|E|+|\vec f|\ge|U|+2|S|.
\end{equation}

Combine Eqs.~\eqref{eq:noisy_flip_lemma_ineq1}, \eqref{eq:noisy_flip_lemma_ineq2} to get
\begin{equation}
(1-\epsilon)b|E|-|U|\le|S|\le\frac12(b|E|+|\vec f|-|U|),
\end{equation}
or, removing $|S|$ entirely and using $\epsilon<\frac14-\frac{r}{b}$,
\begin{equation}\label{eq:noisy_flip_lemma_ineq3}
\left(\frac12b+2r\right)|E|<(1-2\epsilon)b|E|\le|\vec f|+|U|.
\end{equation}
Thus, if $|\vec f|/r\le|E|\le(1-\delta)n$, then
\begin{equation}
|U|>\frac12b|E|+|\vec f|.
\end{equation}
If $u_x=|\Gamma(x)\cap U|$ for $x\in E$ is the number of unsatisfied checks that $x$ participates in, then clearly
\begin{equation}
|\vec f|+\sum_{x\in E}u_x\ge|U|>\frac12b|E|+|\vec f|,
\end{equation}
or, simply,
\begin{equation}
\frac{1}{|E|}\sum_{x\in E}u_x>\frac12b,
\end{equation}
implying that there exists $y\in E$ such that $u_y>b/2$. Thus, there is always a bit to flip in step (2a) provided $|\vec f|/r\le|E|\le(1-\delta)n$. 

We complete the proof by showing that $|E|\le(1-\delta)n$ always holds and therefore the flip algorithm only finishes if $|E|=|\pvec e'-\vec e|<|\vec f|/r$.

The noisy flip algorithm flips one bit at a time and $|E|<\lfloor(1-\delta)n\rfloor$ at the beginning of the algorithm, so if $|E|>(1-\delta)n$ at some time, then there is a prior time at which $|E|=\lfloor(1-\delta)n\rfloor$. Then, we can apply Eq.~\eqref{eq:noisy_flip_lemma_ineq3} to find
\begin{equation}
|U|\ge(1-2\epsilon)b\lfloor(1-\delta)n\rfloor-|\vec f|.
\end{equation}
Let $U_0$ denote $U$ at the very start of the algorithm (i.e.~when $\pvec e'=0^n$ and $|E|=|\vec e|$). By Eq.~\eqref{eq:noisy_flip_lemma_ineq2}, we see
\begin{equation}
|U_0|\le b|\vec e|+|\vec f|.
\end{equation}
Moreover, the intermediate rounds of the algorithm always decrease the size of $U$. So, $|U_0|>|U|$ and hence
\begin{equation}
b|\vec e|+2|\vec f|>(1-2\epsilon)b\lfloor(1-\delta)n\rfloor.
\end{equation}
However, this is in contradiction with Eq.~\eqref{eq:noisy_flip_assumption}.
\end{proof}

We briefly remark that although $(1-2\epsilon)\lfloor(1-\delta)n\rfloor<d/2$ by Lemma~\ref{lem:expander_code_parameters}, it is not much less than the lower bound on $d/2$ from that lemma. The difference is the factor $(1-2\epsilon)/(1-\epsilon)$, which is constant and near unity when $\epsilon$ is constant and small. Also, since $|\vec e|+|\vec f|\ge|\vec e|+\frac2b|\vec f|$, the assumption
\begin{equation}
|\vec e|+|\vec f|\le (1-2\epsilon)\lfloor(1-\delta)n\rfloor
\end{equation}
is a weaker replacement for Eq.~\eqref{eq:noisy_flip_assumption}, but one that makes the total number of errors $|\vec e|+|\vec f|$ more prominent.

This theorem implies that errors can be kept at a manageable level over time. A simple model of data storage is one in which we periodically error correct based on noisy readout of the parity checks, and noise on the data occurs in between these corrections. Suppose at most $|\vec e|$ errors occur on the data between corrections and during correction at most $|\vec f|$ parity checks are misread. Then, after correction, Theorem~\ref{thm:flip_decoder} guarantees at most $|\vec f|/r$ errors remaining on the data. These errors combine with the $|\vec e|$ data errors in the next step. Thus, a steady state is achieved -- following any correction the data has at most $|\vec f|/r$ errors provided that
\begin{equation}
|\vec e|+\frac1r|\vec f|+\frac2b|\vec f|\le(1-2\epsilon)\lfloor(1-\delta)n\rfloor.
\end{equation}
It is sufficient (though weaker) for
\begin{equation}
|\vec e|+2|\vec f|\le(1-2\epsilon)\lfloor(1-\delta)n\rfloor.
\end{equation}
In this classical scenario, assuming constant error rates, $|\vec e|$ and $|\vec f|$ both scale linearly with $n$ and so this condition is realistically achievable, even asymptotically.

\section{Proof of Lemma~\ref{lem:gauge_fixing_distance}}\label{app:gauge_fixing_distance}
A subsystem code's distance is the minimum weight of a dressed logical operator. Thus, to show $D(\mathcal{G}')\ge D(\mathcal{G})$, we just need to show $\hat{\mathcal{L}}(\mathcal{G}')\le\hat{\mathcal{L}}(\mathcal{G})$. As $\mathcal{G}'$ is a gauge-fixing of $\mathcal{G}$, we have that $\mathcal{S}(\mathcal{G})\le\mathcal{S}(\mathcal{G}')\le\mathcal{G}'\le\mathcal{G}$ and $K(\mathcal{G})=K(\mathcal{G}')$.

Notice first that $\mathcal{L}(\mathcal{G})\le\mathcal{L}(\mathcal{G}')$ because anything that commutes with all elements of $\mathcal{G}$ also commutes with all elements of $\mathcal{G}'\le\mathcal{G}$. Second, elements of $\mathcal{S}(\mathcal{G}')-\mathcal{S}(\mathcal{G})$ are not in $\mathcal{L}(\mathcal{G})$, and so the quotient groups $\mathcal{L}(\mathcal{G})/\mathcal{S}(\mathcal{G})$ and $\mathcal{L}(\mathcal{G})/\mathcal{S}(\mathcal{G}')$ are isomorphic. Thus, combine these two observations to get
\begin{equation}
\mathcal{L}(\mathcal{G})/\mathcal{S}(\mathcal{G})=\mathcal{L}(\mathcal{G})/\mathcal{S}(\mathcal{G}')\le\mathcal{L}(\mathcal{G}')/\mathcal{S}(\mathcal{G}').
\end{equation}
However, $K(\mathcal{G})=K(\mathcal{G}')$ dictates that $|\mathcal{L}(\mathcal{G})/\mathcal{S}(\mathcal{G})|=|\mathcal{L}(\mathcal{G}')/\mathcal{G}'|$ so
\begin{equation}\label{eq:equal_quotients}
\mathcal{L}(\mathcal{G})/\mathcal{S}(\mathcal{G})=\mathcal{L}(\mathcal{G}')/\mathcal{S}(\mathcal{G}').
\end{equation}

Because $\mathcal{S}(\mathcal{G})\le\mathcal{G}$ and $\mathcal{S}(\mathcal{G}')\le\mathcal{G}'$,
\begin{align}
\hat{\mathcal{L}}(\mathcal{G})&=\mathcal{G}\text{\space}\mathcal{L}(\mathcal{G})=\mathcal{G}(\mathcal{L}(\mathcal{G})/\mathcal{S}(\mathcal{G})),\\
\hat{\mathcal{L}}(\mathcal{G}')&=\mathcal{G}'\text{\space}\mathcal{L}(\mathcal{G}')=\mathcal{G}'(\mathcal{L}(\mathcal{G}')/\mathcal{S}(\mathcal{G}')).
\end{align}
Using Eq.~\eqref{eq:equal_quotients} and the fact that $\mathcal{G}'\le\mathcal{G}$, we get $\hat{\mathcal{L}}(\mathcal{G}')\le\hat{\mathcal{L}}(\mathcal{G})$, completing the proof.

\bibliography{References}

\end{document}